\documentclass[12pt]{article}
\usepackage{amsfonts,epsfig,latexsym,amscd,amsmath,theorem,mathrsfs}
\usepackage{color}
\newcommand{\R}{{\mathbb{R}}}
\newcommand{\Z}{{\mathbb{Z}}}
\newcommand{\N}{{\mathbb{N}}}
\newcommand{\C}{{\mathbb{C}}}
\newcommand{\I}{{\mathbb{I}}}

\newcommand{\CP}{{\mathbb{C}}{{P}}}
\newcommand{\RP}{{\mathbb{R}{{P}}}}
\newcommand{\beq}{\begin{equation}}
\newcommand{\eeq}{\end{equation}}
\newcommand{\bea}{\begin{eqnarray}}
\newcommand{\eea}{\end{eqnarray}}
\newcommand{\ben}{\begin{eqnarray*}}
\newcommand{\een}{\end{eqnarray*}}
\newcommand{\ra}{\rightarrow}

\newcommand{\cd}{\partial}
\newcommand{\wt}{\widetilde}
\newcommand{\wh}{\widehat}
\newcommand{\less}{\backslash}
\newcommand{\M}{{\sf M}}
\newcommand{\msn}{{\sf M}_n}

\newcommand{\hess}{{\sf Hess}}

\newcommand{\rat}{{\sf Rat}}

\newcommand{\su}{{\mathfrak{su}}}
\newcommand{\so}{{\mathfrak{so}}}

\newcommand{\ric}{{\sf Ric}}
\newcommand{\en}{{\cal E}}

\def \d{\mathrm{d}}
\newcommand{\dstar}{\delta}
\newcommand{\ip}[1]{\langle #1 \rangle}
\newcommand{\ket}[1]{\left| #1\right.\rangle}

\newcommand{\Cas}{\mathscr{C}}
\newcommand{\jac}{\mathscr{J}}
\newcommand{\rr}{\mathscr{R}}
\newcommand{\dd}{\mathscr{D}}
\newcommand{\vol}{{\rm vol}}
\newcommand{\rot}{{\sf R}}

\newcommand{\vphi}{\varphi}

\newcommand{\ol}{\overline}

\newcommand{\xivec}{\mbox{\boldmath{$\xi$}}}
\newcommand{\sigvec}{\mbox{\boldmath{$\sigma$}}}
\newcommand{\lamvec}{\mbox{\boldmath{$\lambda$}}}
\newcommand{\tauvec}{\mbox{\boldmath{$\tau$}}}
\newcommand{\thetavec}{\mbox{\boldmath{$\theta$}}}
\newcommand{\Omegavec}{\mbox{\boldmath{$\Omega$}}}

\newcommand{\kvec}{\mbox{\boldmath{$k$}}}
\newcommand{\vvec}{\mbox{\boldmath{$v$}}}
\newcommand{\Xvec}{\mbox{\boldmath{$X$}}}

\newcommand{\cdv}{\mbox{\boldmath{$\cd$}}}
\newcommand{\vc}[1]{\mbox{\boldmath{$#1$}}}

\newcommand{\tr}{{\rm tr}\, }
\newcommand{\cosec}{{\rm cosec}\, }
\newcommand{\spec}{{\rm spec}\, }
\newcommand{\id}{{\rm Id}}

\theoremstyle{plain}
\newtheorem{thm}{Theorem}
\newtheorem{lemma}[thm]{Lemma}
\newtheorem{prop}[thm]{Proposition}

{\theorembodyfont{\rmfamily}

}
\newcommand{\news}{\setcounter{equation}{0}}
\newenvironment{proof}{\noindent{\it Proof:\, }}{\hfill$\Box$\vspace*{0.5cm}
}

\begin{document}

\title{Quantum lump dynamics on the two-sphere}
\author{S. Krusch\thanks{E-mail: {\tt S.Krusch@kent.ac.uk}} \\
School of Mathematics, Statistics and Actuarial Science, University of 
Kent\\
Canterbury CT2 7NF, England\\ \\
J.M. Speight\thanks{E-mail: {\tt speight@maths.leeds.ac.uk}}\\
School of Mathematics, University of Leeds\\
Leeds LS2 9JT, England}

\date{}
\maketitle
Version: November 1, 2012

\begin{abstract}
It is well known that the low-energy classical 
dynamics of solitons of Bogomol'nyi
type is well approximated by geodesic motion in $\M_n$, the moduli space
of static $n$-solitons. There is an obvious quantization of this dynamics
wherein the wavefunction $\psi:\M_n\ra\C$ evolves according to the
Hamiltonian $H_0=\frac12\triangle$, where $\triangle$ is the Laplacian on 
$\M_n$. Born-Oppenheimer reduction of analogous mechanical 
systems suggests, however,
that this simple Hamiltonian should receive corrections including
$\kappa$, the scalar curvature of $\M_n$, and $\Cas$,
the $n$-soliton Casimir energy, which are usually difficult or impossible
to compute, and whose effect on the energy spectrum is unknown.
This paper analyzes the spectra of $H_0$ and
two corrections to it suggested by work of Moss and Shiiki, namely
$H_1=H_0+\frac14\kappa$ and $H_2=H_1+\Cas$, in the simple but 
nontrivial case of a single $\CP^1$ lump moving on the two-sphere. Here
$\M_1=\rat_1$, a noncompact k\"ahler 6-manifold invariant
under an $SO(3)\times SO(3)$ action, whose geometry
is well understood. The symmetry gives rise to two conserved angular
momenta, spin and isospin. By exploiting the diffeomorphism
$\rat_1\cong TSO(3)$, a hidden isometry of $\rat_1$
is found 
which implies that all three energy spectra are symmetric under
spin-isospin interchange. The Casimir energy is found exactly
on an $SO(3)$ submanifold of  $\rat_1$, using standard results from 
harmonic map theory and zeta function regularization, and 
approximated numerically on the rest of $\rat_1$. The lowest 19
eigenvalues of $H_i$ are found, and their spin-isospin and parity
compared for $i=0,1,2$. It is found that the curvature corrections in 
$H_1$ lead to a qualitatively unchanged low-level spectrum while the 
Casimir energy in $H_2$ leads to significant changes. 
The scaling behaviour of the spectra under changes in the radii of the
domain and target spheres is analyzed, and it is found that the
disparity between the spectra of $H_1$ and $H_2$ is reduced when the
target sphere is made smaller.
\end{abstract}

\maketitle

\section{Introduction}
\label{sec:intro}
\news

Many field theories arising naturally in theoretical high energy physics may
be said to be of Bogomol'nyi type. For such theories there is a topological
lower bound on the energy of field configurations, and this bound is attained
only by solutions of a first order ``self-duality'' equation, the so-called
solitons of the theory. The solitons are stable by virtue of their
energy-minimizing property, and are generically spatially localized lumps of
energy with strongly particle-like characteristics. 
Examples are magnetic monopoles, abelian Higgs vortices and sigma model lumps.
There is a well-developed geometric framework for studying the classical
low-energy 
dynamics of such solitons, proposed originally for monopoles by Manton 
\cite{man1}, called the geodesic approximation.
Here, $n$-soliton trajectories are approximated by
geodesics in the moduli space $\msn$ of static $n$-solitons, with respect to
the metric induced by the kinetic energy functional of the field theory
(usually called the $L^2$ metric). This approach has been extremely fruitful, 
providing both important qualitative
insight into topological soliton dynamics and good agreement with
numerical analysis of the full field theory. In the case of vortices and
monopoles, the geodesic approximation is backed by rigorous analysis
\cite{stu1,stu2}. For a comprehensive review, see \cite{mansut}.

When one comes to quantize the low-energy 
dynamics of such solitons, two approaches
are possible. Since geodesic motion on $\msn$ captures the classical
soliton dynamics so well, it is natural simply to quantize that
\cite{gibman}. Then
the quantum $n$-soliton state is specified by a wavefunction $\psi:\msn\ra\C$,
evolving subject to the Hamiltonian
\beq\label{mant}
H_{geo}=\tfrac12\triangle,
\eeq
where $\triangle$ is the Hodge Laplacian on $\msn$. This has the virtue of 
simplicity, but it ignores the degrees
of freedom normal to the moduli space completely. An alternative is to
make a low energy reduction of the full quantum field theory by means
of the Born-Oppenheimer approximation. This has been carried out
for sigma model lumps by Moss and Shiiki \cite{mosshi}. Arguing by analogy
with finite dimensional mechanical systems, they find that, once again,
the low energy quantum dynamics of $n$ solitons can be described by a 
wavefunction on $\msn$, but that the Hamiltonian is
\beq\label{mosh}
H_{BO}=\tfrac12\triangle+\tfrac14\kappa-\tfrac18\|\kvec\|^2+U+\cdots,
\eeq
where $\kappa$ is the scalar curvature of $\msn$, $\kvec$ is the mean 
curvature of the embedding of $\msn$ into the (infinite dimensional)
field configuration space 
and $U$ is a potential on
$\msn$ incorporating the residual effects of the normal modes. This
potential is rather complicated, but its dominant term is the 
$n$-soliton Casimir
energy, that is, the total zero-point energy of the normal modes to the
static $n$-soliton.

The aim of the current paper is to 
compare the spectra of $H_{geo}$ and $H_{BO}$, to determine to
 what extent the extra terms
in (\ref{mosh}) change the quantum energy spectrum. 
Of course, one expects the numerical
values of the energy eigenvalues to change to some extent. In itself, this
is not particularly important. We will be more interested in the
{\em qualitative} features of the energy spectrum. For example, on 
a noncompact moduli space, one could easily imagine that $H_{geo}$ may have
no bound states, while $H_{BO}$ does. In the extreme case, one could find that
the spectrum of $H_{geo}$ is continuous while that of $H_{BO}$ is discrete.
Clearly, $H_{geo}$ would be a disastrously bad approximation in this case.
Other, more refined, qualitative features can also be compared, for example,
the dimension and symmetry properties of the eigenspaces of $H_{geo}$ and
$H_{BO}$, ordered by energy. 

We address this issue for a specific example, namely a single $\CP^1$ 
lump moving on the two sphere (spacetime $S^2\times\R$). This choice
has several mathematical advantages, making it simple enough
to be tractable, but not too simple to provide 
a nontrivial test of the truncation $H_{geo}$.
 First, the static $n$-solitons are
explicitly known -- they are rational maps, so $\msn=\rat_n$. Second, by
choosing physical space to be compact ($S^2$ here) rather than $\R^2$
(as for the model on Minkowski space $\R^{2+1}$), we ensure that the
$L^2$ metric on $\rat_n$ is well defined (there are no non-normalizable
zero modes). Third, we may think of the $L^2$ metric (formally) as the 
induced metric on $\rat_n$ as a complex submanifold of $(\CP^1)^{S^2}$, 
the infinite dimensional space
of maps $S^2\ra\CP^1$. This latter is (formally) an infinite product of
k\"ahler manifolds, so is k\"ahler, so, as noted by Moss and Shiiki 
\cite{mosshi},
the very awkward extrinsic curvature term is absent, $\|\kvec\|^2=0$.
Fourth,
the $L^2$ metric on $\rat_1$ 
is highly symmetric, so the $L^2$ metric and its scalar curvature can be
computed explicitly. On the other hand, $\rat_1$ is not {\em too} symmetric:
the scalar curvature and Casimir energy are non-constant functions of a single
variable, the lump width, so that (in contrast to the case of monopoles 
and
vortices) even the $n=1$ sector provides a nontrivial test.

To be precise, we will compute the spectra of three approximations to
$H_{BO}$, 
\beq
\label{Hamiltonians}
H_0=H_{geo}=\tfrac12\triangle,\qquad
H_1=\tfrac12\triangle +\tfrac14\kappa,\qquad
H_2=\tfrac12\triangle+\tfrac14\kappa+\Cas,
\eeq
where $\Cas$ is the Casimir energy. We will find that the spectrum 
of $H_0$ can be thought of as a perturbed version of the spectrum of the
Laplacian on $\CP^3$ (equipped with the Fubini-Study metric). Quantum lumps
possess two integer conserved angular momentum quantum numbers, which we call
isospin $k$ and spin $s$, associated with the rotational 
symmetries of the target space $\CP^1$ and the domain $S^2$ respectively. 
We will show that for each choice of $k,s$, the spectral problem
for $H_i$, $i=0,1,2$, reduces to a matrix Sturm-Liouville problem of 
dimension $2\min\{k,s\}+1$. We find that this problem is symmetric under 
interchange of $k$ and $s$, so the energy spectra of $H_0,H_1,H_2$ all
possess this symmetry. A careful analysis of the boundary conditions
for the Sturm-Liouville problem shows that the boundary conditions appropriate
for $H_1$ are the same as for $H_2$, but different from $H_0$: including
scalar curvature changes the boundary conditions. 
Nonetheless, the spectrum is discrete in all three cases. We have
computed (numerically) the first 19 energy eigenvalues for each
Hamiltonian and tabulated $\{k,s\}$ for the eigenstates in order of
increasing energy. The spectra of $H_0$ and $H_1$
are remarkably similar. The order of states is the
same apart from two transpositions. The Casimir energy leads to 
significant changes in the spectrum. However, the effect of the Casimir 
energy becomes less important if the radius of the target space is 
decreased.

\section{The one-lump moduli space}
\news

Identifying $\CP^1\cong S^2$ and using complex stereographic coordinates
$z=x+iy$ and $W$ on domain and codomain respectively, we may identify a map
$\phi:S^2\ra \CP^1\cong S^2$ with a complex function $W(z)$. 
The kinetic and potential
energy functionals of the $\CP^1$ model on $S^2$ are then
\beq
T=\frac12\int 
\frac{4|W_t|^2}{(1+|W|^2)^2}\frac{4{\rm d}x{\rm d}y}{(1+|z|^2)^2},\quad
V=\frac12\int \frac{4(|W_x|^2+|W_y|^2)}{(1+|W|^2)^2}{\rm d}x {\rm d}y
\eeq
respectively. Note that, in contrast with earlier work
\cite{bap,macspe,spe1,spe3} we have given
both domain and codomain spheres the round metric of radius $1$. This will 
allow us to make direct use of results in the harmonic maps literature
when we come to compute the Casimir energy of a lump. We
will consider how our results change when the radii of the domain
and target spheres are altered in section \ref{scaling}.

A theorem of Lichnerowicz \cite{lic} 
shows that, if $\phi$ has topological degree $n$ (assumed,
without loss of generality, to be non-negative), then $V\geq 4\pi n$,
with equality if and only if $\phi$ is holomorphic. So degree $n$ holomorphic
maps $\phi:S^2\ra S^2$ minimize $V$ within their homotopy class: these are
the solitons of the model, and the ``self-duality'' equation is the
Cauchy-Riemann equation. In terms of $W$ and $z$, degree $n$ holomorphic
maps $S^2\ra S^2$ are rational maps of (algebraic) degree $n$, that is,
\beq
W=\frac{a_1z^n+\cdots+a_{n+1}}{a_{n+2}z^n+\cdots+a_{2n+2}},
\eeq
where $a_i$ are $2n+2$ complex constants, $a_1$ and $a_{n+2}$ do not both
vanish, and the numerator and denominator have no common roots. 
One interprets this physically as a static superposition of $n$ lumps.
So the moduli space of static $n$-lumps is $\rat_n$, the space of
degree $n$ rational maps. 

The $L^2$
metric on $\rat_n$ is defined by restricting the kinetic energy functional
$T$ to fields $W(t,z)$ which at each fixed $t$ are degree $n$ rational
maps. Explicitly, in the chart where $a_1\neq 0$, one can define
local complex coordinates $q_i=a_i/a_1$, $i=2,\ldots,2n+2$ on $\rat_n$. Then
allowing $q_2,\ldots,q_{2n+2}$ to vary with time, one substitutes
\beq
W(t,z)=\frac{z^n+q_2(t)z^{n-1}+\cdots+q_{n+1}(t)}
{q_{n+2}(t)z^n+q_{n+3}(t)z^{n-1}+\cdots+q_{2n+2}(t)}
\eeq
into $T$, to obtain 
\beq\label{marmes}
T=\frac12\sum_{ij}\gamma_{ij}\dot{q}_i\ol{\dot{q_j}},\qquad
\gamma_{ij}=\int \frac{4}{(1+|W|^2)^2}
\frac{\cd W}{\cd q_i}\ol{\frac{\cd W}{\cd q_j}}
\frac{4{\rm d}x{\rm d}y}{(1+|z|^2)^2}.
\eeq
The $L^2$ metric is $\gamma=\sum_{i,j}\gamma_{ij}{\rm d}q_i {\rm d}\ol{q}_j$.
See \cite{spe3} for a coordinate free definition of $\gamma$. 
The metric $\gamma$ is manifestly Hermitian. In fact, it is
k\"ahler \cite{rub,spe3}. It is known to be geodesically
incomplete \cite{sadspe}, so the classical geodesic approximation predicts
lumps may collapse to infinitely narrow spikes in finite time.
Numerical simulation and rigorous analysis (albeit on domain $\C$)
confirm that lump collapse can occur, though the geodesic approximation 
gets the fine detail of the singularity formation process wrong
\cite{bizchmtab,linsad,rodste}. 

There is an
isometric action of $G=SO(3)\times SO(3)$ on $\rat_n$, induced by the natural
$SO(3)$ actions on the domain and target spheres. On $\rat_1$ this action has
cohomogeneity 1 (generic $G$ orbits have codimension 1), and, in fact,  almost
completely determines $\gamma$. Consequently, an explicit formula for
$\gamma$ is known in this case, and the geometry is particularly well
understood. For $n=1$, the no common roots condition on the rational map
$W(z)=(a_1z+a_2)/(a_3z+a_4)$ is $a_1a_4-a_2a_3\neq0$, so we may identify each
map with a projective equivalence class $[L]$
of $GL(2,\C)$ matrices,
\beq
\frac{a_1z+a_2}{a_3z+a_4}\leftrightarrow
\left[\:\left(\begin{array}{cc}a_1&a_2\\a_3&a_4\end{array}\right)\:\right].
\eeq
 Hence
$\rat_1\cong PL(2,\C)$. 
The action of $G$ on $\rat_1$ corresponds, under this identification, with the
natural action of $PU(2)\times PU(2)$ on $PL(2,\C)$:
\beq
([U_1],[U_2]):[L]\mapsto [U_1LU_2^{-1}].
\eeq

Now every $[L]\in PL(2,\C)$ has a unique polar decomposition
\beq
[L]=[U(\Lambda\I_2+\lamvec\cdot\tauvec)],
\eeq
where $([U],\lamvec)\in PU(2)\times\R^3$,
$\Lambda=\sqrt{1+\lambda^2}$,
$\lambda=|\lamvec|$, and $\tau_1,\tau_2,\tau_3$ are the
Pauli spin matrices. Hence $\rat_1\cong PU(2)\times\R^3$.
Having chosen a basis $\{\frac{i}{2}\tau_a\}$ for $\su(2)$, we have a
canonical identification $PU(2)\cong SO(3)$ under which $[U]$ is
identified with the orthogonal transformation $Ad_U:\su(2)\ra\su(2)$.
Hence $\rat_1\cong SO(3)\times\R^3$, and the $G$ action is
\beq
(R_1,R_2):(R,\lamvec)\mapsto(R_1RR_2^{-1},R_2\lamvec).
\eeq
From this we see that the $G$-orbits are level sets of $\lambda$, generically
diffeomorphic
to $SO(3)\times S^2$ (when $\lambda>0$), the only exception being
$\lambda=0$,
which is diffeomorphic to $SO(3)$.
Physically, the lump corresponding to $(R,\lamvec)\in SO(3)\times\R^3$ 
has maximum
energy density
 at $-\lamvec/\lambda\in S^2$, sharpness proportional to $\lambda$ and
internal orientation $R$. The $\lambda=0$ lumps have uniform energy
density. 

The following explicit characterization of $G$ invariant
k\"ahler metrics on $\rat_1$ was established in \cite{spe3} (see
\cite{bap} for an alternative viewpoint, exploiting more directly the
covering $SL(2,\C)\ra\rat_1$):

\begin{prop}\label{metprop} Let $\gamma$ be an $SO(3)\times SO(3)$ 
invariant 
k\"ahler metric on $\rat_1$. Then
\beq
\gamma=A_1\, {\rm d}\lamvec\cdot {\rm d}\lamvec+
A_2(\lamvec\cdot {\rm d}\lamvec)^2+
A_3\, \sigvec\cdot\sigvec+A_4(\lamvec\cdot\sigvec)^2+
A_1\lamvec\cdot (\sigvec\times {\rm d}\lamvec),
\eeq
where $A_1,\ldots,A_4$ are smooth functions of $\lambda$ only, all determined
from the single function $A_1=A(\lambda)$ by the relations
\beq
A_2=\frac{A(\lambda)}{1+\lambda^2}+\frac{A'(\lambda)}{\lambda},
\quad A_3=\frac{1}{4}(1+2\lambda^2)A(\lambda),
\quad A_4=\frac{1}{4\lambda}(1+\lambda^2)A'(\lambda).
\eeq
Here $\sigma_1,\sigma_2,\sigma_3$ are the left invariant one forms on
$SO(3)$ dual to the basis $\{\frac{i}{2}\tau_a:a=1,2,3\}$ for $\su(2)\cong
\so(3)$, $\times$ and $\cdot$ denote the vector and scalar product on $\R^3$
respectively and juxtaposition of one-forms denotes symmetrized tensor
product. 
\end{prop}

So symmetries and the k\"ahler property determine the metric up to a single 
function $A(\lambda)$, which we may think of as the squared
length of the vector $\cd/\cd\lambda_1$ at the point 
$(\I_3,(0,0,\lambda))\in SO(3)\times\R^3$, corresponding to the
rational map 
\beq
W(z)=\mu z,\qquad
\mu=\frac{\Lambda+\lambda}{\Lambda-\lambda}.
\eeq
For the $L^2$ metric, one finds that
\beq\label{Adef}
A=\frac{32\pi\mu[\mu^4-4\mu^2\log\mu-1]}{(\mu^2-1)^3}.
\eeq
It follows from these formulae that $(\rat_1,\gamma)$ has finite diameter and
volume ($16^3\pi^6/6$ according to Baptista \cite{bap}), 
is Ricci positive and has unbounded scalar and holomorphic
sectional curvatures. 
Examining the large $\lambda$ behaviour of $\gamma$,
one finds that the boundary at infinity of $\rat_1$ is
$S^2\times S^2$. 
Geodesic flow in $(\rat_1,\gamma)$ was studied in detail in
\cite{spe1,hasspe}, and 
turns out to be surprisingly complicated given the homogeneity and isotropy
of physical space. In particular, lumps generically
do not travel along great circles in $S^2$.

It is important to realize that {\em any} k\"ahler metric on $\rat_1$
invariant under the $SO(3)\times SO(3)$ action
must have the structure of Proposition \ref{metprop} for some function
$A(\lambda)$. There is one other natural metric on $\rat_1$ which we
will have reason to consider. By identifying a rational map with the
projective equivalence class of its coefficients, we may think of $\rat_1$
as an open subset of $\CP^3$,
\beq
\rat_1\hookrightarrow\CP^3,\qquad
\frac{a_1z+a_2}{a_3z+a_4}\mapsto [a_1,a_2,a_3,a_4].
\eeq
We equip $\CP^3$ with the Fubini-Study metric of constant holomorphic
sectional curvature $4$, whence $\rat_1$ inherits a k\"ahler metric,
which we shall also call the Fubini-Study metric $\gamma_{FS}$. 
It is not hard to show \cite{spe3} that $\gamma_{FS}$ is $SO(3)\times SO(3)$
invariant, and so has the structure of Proposition \ref{metprop}. One
finds that the coefficient function is
\beq
\label{AFS}
A_{FS}=\frac{2\mu}{1+\mu^2} = \frac{1}{2\lambda^2+1}.
\eeq 
Finally, given any $SO(3)\times SO(3)$ invariant k\"ahler metric on
$\rat_1$, it is convenient to define a second metric coefficient
function, the squared length at $(\I_3,(0,0,\lambda))$ of $\theta_3$,
where $\theta_1,\theta_2,\theta_3$ are the left-invariant vector fields on
$SO(3)$ dual to $\sigma_1,\sigma_2,\sigma_3$.
This turns out to be
\beq\label{Bdef}
B=\frac14(\lambda^2+\Lambda^2)A
+\frac14{\lambda\Lambda^2}A'.
\eeq
 For the $L^2$ and Fubini-Study
metrics one finds, respectively,
\bea
B&=&\frac{32\pi\mu^2}{(\mu^2-1)^3}[(\mu^2+1)\log\mu-\mu^2+1],\nonumber \\
B_{FS}&=&\frac{1}{4(\lambda^2+\Lambda^2)^2} = 
\frac{1}{4(2\lambda^2+1)^2}.
\eea

\section{The Laplacian on $\rat_1$}
\label{LRat1}
\news

We begin the computation of the Laplacian on functions on $\rat_1$ by
proving a lemma which generalizes the well-known expression for $\triangle$ in
local coordinates.

\begin{lemma}\label{laplem}
Let $(M^m,g)$ be a Riemannian manifold of dimension $m$, $\{X_i\}$ be a 
local frame on $M$, $\{\nu_i\}$ be the associated coframe,
and $\{\mu_i\}$ be the associated basis for $\Lambda^{m-1}M$, that is,
$$
\mu_1=\nu_2\wedge\nu_3\wedge\cdots\wedge\nu_m,\quad
\mu_2=\nu_1\wedge\nu_3\wedge\cdots\wedge\nu_m,\quad\cdots,\quad
\mu_m=\nu_1\wedge\nu_2\wedge\cdots\wedge\nu_{m-1}.
$$
If all the $(m-1)$-forms $\mu_i$ are closed, then the Laplacian on
functions is
$$
\triangle f=-\frac{1}{\sqrt{|g|}}
\sum_{i,j}X_i[\sqrt{|g|}\wh{g}(\nu_i,\nu_j)X_j[f]],
$$
where $\wh{g}$ is the inverse metric and $|g|=\det(g_{..})$.
\end{lemma}

\begin{proof}
We use the summation convention on repeated indices, and define $g^{ij}
=\wh{g}(\nu_i,\nu_j)$. 
The Laplacian on functions is $\triangle=-*{\rm d}*{\rm d}$, where $*$
denotes the Hodge isomorphism, so
\bea
\triangle f&=&-*{\rm d}(X_i[f]*\nu_i)
=-*{\rm d}(X_i[f]\frac{\sqrt{|g|}}{(m-1)!}g^{ij}\epsilon_{ji_2i_3\cdots i_m}
\nu_{i_2}\wedge\cdots\wedge\nu_{i_m})\nonumber\\
&=&-*X_k[X_i[f]\sqrt{|g|}g^{ij}]\frac{\epsilon_{ji_2i_3\cdots i_m}}{(m-1)!}
\nu_k\wedge\nu_{i_2}\wedge\cdots\wedge\nu_{i_m}\nonumber\\
&=&-*X_j[X_i[f]\sqrt{|g|}g^{ij}]\nu_{1}\wedge\nu_{2}\wedge\cdots\wedge\nu_{m},
\nonumber
\eea
where we have used ${\rm d}(\nu_{i_2}\wedge\cdots\wedge\nu_{i_m})=0$. Now the
volume form is ${\rm vol}=\sqrt{|g|}\nu_1\wedge\cdots\wedge\nu_m$ and
$*{\rm vol}=1$, so
$$
\triangle f=-X_j[X_i[f]\sqrt{|g|}g^{ij}]\frac{1}{\sqrt{|g|}}
$$
as claimed.
\end{proof}

Note that any local coordinate basis $X_i=\cd/\cd x^i$ satisfies the
conditions automatically, and that the formula for $\triangle$
reduces to the usual
expression in this case. For our purposes it is convenient to use
the (global) frame $\{\cd_a, \theta_a\: :\: a=1,2,3\}$ on
$\rat_1$, where $\cd_a = \cd/\cd\lambda_a$ and $\theta_a$ are the 
left-invariant vector fields on 
$SO(3)$ dual to $\sigma_a$. 
We thus require an expression for the inverse metric
$\wh{\gamma}$ relative to this frame.
 
\begin{prop}\label{invmetprop}
Let $\gamma$ be an $SO(3)\times SO(3)$ invariant k\"ahler metric
on $\rat_1$ determined, as in Proposition \ref{metprop}, by the function
$A(\lambda$). Then the inverse metric $\wh{\gamma}$ is
\beq
\wh{\gamma}=C_1\cdv\cdot\cdv+C_2(\lamvec\cdot\cdv)^2+
C_3\thetavec\cdot\thetavec+C_4(\lamvec\cdot\thetavec)^2+
C_5\lamvec\cdot(\thetavec\times\cdv)
\eeq
where $C_1,\ldots,C_5$ are smooth functions of $\lambda$ alone satisfying
$$
C_1=\frac{\Lambda^2+\lambda^2}{\Lambda^2A},\quad
C_1+\lambda^2C_2=\frac{\Lambda^2}{4B},\quad
C_3=\frac{4}{\Lambda^2A},\quad
C_3+\lambda^2C_4=\frac{1}{B},\quad
C_5=-C_3.
$$
Here $\cdv=(\cd/\cd\lambda_1,\cd/\cd\lambda_2,\cd/\cd\lambda_3)$,
$\thetavec=(\theta_1,\theta_2,\theta_3)$, juxtaposition of
vector fields denotes symmetrized tensor product, and $B$ is determined by $A$
as in equation (\ref{Bdef}).
\end{prop}

\begin{proof}
The same symmetry argument used to prove Proposition \ref{metprop}
(Proposition 3.1 of \cite{spe3}) shows that every invariant symmetric
$(2,0)$ tensor is of the above form, with $C_1,\ldots,C_5$ some functions of 
$\lambda$ only. The formulae for $C_1,\ldots,C_5$ result from
explicit computation of $\|{\rm d}\lambda_a\|^2$, $\|\sigma_a\|^2$ and
$\langle {\rm d}\lambda_a,\sigma_b\rangle$
$a,b=1,2,3$, at the
point $(\I_3,(0,0,\lambda))$ for general $\lambda$, using the unitary
frame $\{e_a,Je_a\: :\: a=1,2,3\}$ introduced in section 4.1 of \cite{spe3}.
For example, at $(\I_3,(0,0,\lambda))$,
$$
\|\sigma_3\|^2=C_3+\lambda^2 C_4=\sum_{a=1}^3(\sigma_3(e_a)^2+
\sigma_3(Je_a)^2)=\sigma_3(Je_3)^2=\frac{1}{B}.
$$
%and this formula holds globally on $\rat_1$ by $SO(3)\times SO(3)$
%invariance. 
\end{proof}

\begin{prop}\label{lapprop}
Let $\gamma$ be an $SO(3)\times SO(3)$ invariant k\"ahler metric
on $\rat_1$ determined as in Proposition \ref{metprop} by the function 
$A(\lambda)$. Then the Laplacian on $(\rat_1,\gamma)$ is
\bea
\triangle f&=&-\frac{4}{\Lambda^2A}\left\{\thetavec\cdot\thetavec f+
\vc{\lambda}\cdot(\cdv\times\thetavec)f-\frac{1}{\lambda^2}\left[1-
\frac{\Lambda^2A}{4B}\right](\lamvec\cdot\thetavec)^2f\right\}\nonumber \\
&&-\frac{1}{\Lambda\lambda^2A^2B}\frac{\cd\:}{\cd\lambda}\left(
\frac{\lambda^2\Lambda^3A^2}{4}\frac{\cd f}{\cd\lambda}\right)-
\frac{\Lambda^2+\lambda^2}{\lambda^2\Lambda^2A}(\lamvec\times\cdv)\cdot
(\lamvec\times\cdv)f.\label{lap}
\eea
\end{prop}

\begin{proof} 
We work with the frame $\{\cd_a,\theta_a\: :\: a=1,2,3\}$.
Note that ${\rm d}({\rm d}\lambda_a)=0$ and
$$
{\rm d}\sigma_1=\sigma_2\wedge\sigma_3,\quad
{\rm d}\sigma_2=\sigma_3\wedge\sigma_1,\quad
{\rm d}\sigma_3=\sigma_1\wedge\sigma_2,
$$
so any wedge product of five of these one forms is closed. Hence
this frame satisfies the conditions of Lemma \ref{laplem} which, with
Proposition \ref{invmetprop}, immediately gives
\bea
\triangle f&=&-\frac{1}{\sqrt{|\gamma|}}\left\{\cdv\cdot(\sqrt{|\gamma|}(C_1
\cdv f+C_2\lamvec(\lamvec\cdot\cdv[f])))+
\cdv\cdot(-\sqrt{|\gamma|}\frac{C_5}{2}\lamvec\times\thetavec[f])
\right.\nonumber \\
&&\label{huh1}\left.+
\thetavec \cdot (\sqrt{|\gamma|}\frac{C_5}{2}\lamvec\times\cdv[f])+
\thetavec \cdot (\sqrt{|\gamma|}(C_3\thetavec[f] +
C_4\lamvec(\lamvec\cdot\thetavec[f])))\right\}.  
\eea
Now the volume form on $(\rat_1,\gamma)$ is \cite{spe3}
\beq
\label{vol}
{\rm vol}=\frac{\Lambda}{2}A^2B
{\rm d}\lambda_1\wedge {\rm d}\lambda_2\wedge {\rm d}\lambda_3\wedge
\sigma_1\wedge\sigma_2\wedge\sigma_3,
\eeq
whence one sees that $\sqrt{|\gamma|}=\frac{\Lambda}{2}A^2B$ for this frame.
Substituting this, and the formulae for $C_1,\ldots,C_5$ (Proposition 
\ref{invmetprop}) into equation (\ref{huh1}) yields, after some 
straightforward manipulation, the formula claimed.
\end{proof}

As a check on our formula, we should verify that it is 
consistent with the $SO(3)\times SO(3)$ symmetry of $\gamma$. That is,
the operator $\triangle$ must commute with all
Killing vector fields on $(\rat_1,\gamma)$. There are six
independent Killing vector fields on $(\rat_1,\gamma)$,
three generating the left $SO(3)$ action and three
 generating the right $SO(3)$
action. Recall that the left action on $\rat_1\cong SO(3)\times\R^3$
acts by left translation on $SO(3)$ and acts trivially on $\R^3$. Now
left translations on a Lie group are generated by {\em right} invariant
vector fields. So, let $\xi_a$, $a=1,2,3$, be the right invariant vector
fields on $SO(3)$ coinciding at $\I_3$ with $\theta_a$. Since the left
and right actions commute, $[\theta_a,\xi_b]=0$ for all $a,b$ so we
see immediately that $\triangle$ as in formula (\ref{lap})
commutes with each $\xi_a$.
The right action of $SO(3)$ on $\rat_1\cong SO(3)\times\R^3$
acts by right translation on $SO(3)$ and by the fundamental representation
on $\R^3$. Hence the Killing vector fields generating this action are
$X_a=\theta_a+\Phi_a$ where $\Phi_a=\epsilon_{abc}\lambda_b\cd_c$.
Note that $\{\theta_a\}$, $\{\Phi_a\}$ and $\{X_a\}$ each satisfy the
angular momentum algebra, 
\beq
[\theta_a,\theta_b]=-\epsilon_{abc}\theta_c,\quad
[\Phi_a,\Phi_b]=-\epsilon_{abc}\Phi_c,\quad
[X_a,X_b]=-\epsilon_{abc}X_c,
\eeq
and that $[\theta_a,\Phi_b]=0$,
so 
\beq
[X_a,\Xvec\cdot\Xvec]=
[X_a,\thetavec\cdot\thetavec]=
[X_a,\vc{\Phi}\cdot\vc{\Phi}]=0.
\eeq
Clearly all these
vector fields annihilate functions of $\lambda$ and commute with 
$\cd/\cd\lambda$.
It follows that $X_a$ commutes with the first, fourth and fifth terms of
(\ref{lap}). To deal with the second term, note that
\beq
[X_a,\lamvec\cdot(\cdv\times\thetavec)]=
[X_a,\thetavec\cdot\vc{\Phi}]=
\frac12[X_a,\Xvec\cdot\Xvec-\thetavec\cdot\thetavec-\vc{\Phi}\cdot\vc{\Phi}]
=0.
\eeq
Finally, $X_a$ commutes with the third term since
\beq
\label{Xcommutes}
[X_a,\lamvec\cdot\thetavec]=
[\theta_a+\epsilon_{abc}\lambda_b\cd_c,\lambda_d\theta_d]=
\lambda_d(-\epsilon_{ade}\theta_e)+\epsilon_{abd}\lambda_b\theta_d=0.
\eeq

\section{Casimir energy}
\news

The dominant term in Moss and Shiiki's Hamiltonian arising from the 
degrees of freedom orthogonal to the moduli space is the Casimir energy 
$\Cas,$ which is formally given by
\beq
\label{Cformal}
\Cas=\frac12\sum_i\omega_i,
\eeq
where $\omega_i$ are the frequencies of oscillation of the normal modes
of the static solution. In the field theory 
context, the above sum is infinite and divergent, so we must regularize it 
in some way. We return
to this issue below. First we set about computing the frequencies $\omega_i$
in the case of interest, a single $\CP^1$ lump on $S^2$. To achieve this,
we will make use of some standard results in the stability theory of harmonic
maps, which we begin by briefly reviewing. This material is treated in detail
in \cite[ch.\ 5]{ura}.

\subsection{The spectrum of the Jacobi operator}

Recall that given a map $\vphi:(M,g)\ra (N,h)$ between Riemannian manifolds,
its Dirichlet energy is 
\beq
E=\frac12\int_M\|\d\vphi\|^2,
\eeq
 and the map is 
harmonic if it is a critical point of $E$. This is directly relevant to us 
since, with the choice $(M,g)=(N,h)=$ the round sphere of radius $1$,
$E$ coincides precisely with $V$, the $\CP^1$ model's potential energy 
functional. Hence, $\CP^1$ lumps are harmonic maps $S^2\ra S^2$. Given a
harmonic map $\vphi:M\ra N$, one defines its Hessian,
a symmetric bilinear form on
$\Gamma(\vphi^{-1}TN)$ (the space of smooth sections of the vector bundle
over $M$ whose fibre over $p\in M$ is the tangent space $T_{\vphi(p)}N$), 
as follows. Let $\vphi_{s,t}:M\ra N$ be
a smooth two-parameter variation of $\vphi$ (so $\vphi_{0,0}=\vphi$), with
$\cd_s\vphi_{s,t}|_{s=t=0}=X$, 
$\cd_t\vphi_{s,t}|_{s=t=0}=Y\in\Gamma(\vphi^{-1}TN)$. Then
\beq
\hess_\vphi(X,Y)=\left.\frac{\cd^2\:\:}{\cd s\, \cd
  t}E(\vphi_{s,t})\right|_{s=t=0}. 
\eeq
One says that $\vphi$ is {\em stable} if $\hess_\vphi(X,X)\geq 0$ for
all $X$.  
Clearly, if $\vphi$ minimizes $E$ in its homotopy class, as in our case, it
is stable.

Associated with $\hess_\vphi$, there is a self-adjoint, elliptic, linear 
differential operator 
$\jac_\vphi:\Gamma(\vphi^{-1}TN)\ra\Gamma(\vphi^{-1}TN)$, which is known 
as the Jacobi operator and is defined by
\beq
\hess_\vphi(X,Y)=\ip{X,\jac_\vphi Y}_{L^2}=\int_Mh(X,\jac_\vphi Y).
\eeq
If $M$ is compact,
the harmonic map $\vphi$ is stable if and only if the spectrum of $\jac_\vphi$
is non-negative, and this spectrum is discrete, each eigenvalue having finite 
multiplicity. Any map $\psi:M\ra N$ which is sufficiently close pointwise
to a harmonic map $\vphi$
can be uniquely written $\psi=\exp_{\vphi}X$, where $\|X\|$ is
pointwise small and $\exp_\vphi$ denotes the geodesic exponential
map based at $\vphi$ (explicitly, given $X\in\Gamma(\vphi^{-1}TN)$, $\exp_\vphi X$ is the map $M\ra N$ which sends each $p\in M$ to 
$y(1)$, where $y(t)$ is the geodesic in $N$ with initial data $y(0)=\vphi(p)$, $\dot{y}(0)=X(p)$). Then
\beq
E(\psi)=E(\vphi)+\frac12\hess_\vphi(X,X)+O(X^3)=
E(\vphi)+\frac12\ip{X,\jac_\vphi X}_{L^2}+O(X^3),
\eeq
whence it is clear that the eigenvalues of $\jac_\vphi$ are precisely
$\omega_i^2$, the squared frequencies we require to compute $\Cas$.

There is an explicit formula for the Jacobi operator of a general harmonic
map $\vphi:(M,g)\ra (N,h)$,
\beq
\jac_\vphi=\triangle_\vphi-\rr_\vphi,
\eeq
where $\triangle_\vphi$ is the rough Laplacian on $\Gamma(\vphi^{-1}TN)$, and
$\rr_\vphi$ is a certain section of ${\rm End}(\vphi^{-1}TN)$ constructed
from
the curvature tensor $R^N$ on $N$. Explicitly, given a choice of local
orthonormal frame $E_1,E_2,\ldots,E_m$ on $M$,
\bea
\triangle_\vphi Y&=&-\tr\nabla^\vphi\nabla^\vphi Y=-
\sum_{i=1}^m\left\{\nabla_{E_i}^\vphi(\nabla_{E_i}^\vphi Y)
-\nabla^\vphi_{\nabla^M_{E_i}E_i}Y\right\},
\\
\rr_\vphi Y&=&\sum_{i=1}^mR^N(Y,\d\vphi E_i)\d\vphi E_i,
\eea
where $\nabla^M,\nabla^N$ are the Levi-Civita connexions of $M,N$
respectively, and $\nabla^\vphi$ is the pullback to $\vphi^{-1}TN$ of
$\nabla^N$. In the case of interest to us, $\rr_\vphi$ is somewhat easier
to handle than $\triangle_\vphi$, owing to the following proposition:

\begin{prop}\label{lidl}
Let $\vphi:(M^n,g)\ra (N^n,h)$ be a weakly conformal
mapping between Riemannian manifolds of equal dimension, and
$(N^n,h)$ be Einstein with scalar curvature $\kappa$ (necessarily
constant). Then
$$
\rr_\vphi=\frac{2\kappa}{n^2}\en\, \id,
$$
where $\en=\frac12\|\d\vphi\|^2\in C^\infty(M)$ is the Dirichlet
energy density of $\vphi$.
\end{prop}

\begin{proof} Recall that $\vphi$ is weakly conformal
if there exists a smooth function $f:M\ra\R$ such that
$h(\d\vphi X,\d\vphi Y)=f^2g(X,Y)$ for all vector fields $X,Y$ on $M$.
Let $E_1,\ldots,E_n$ be a local orthonormal frame on $M$.
At all points $p\in M$ where
$f(p)=0$, $\d\vphi=0$ so the desired equality holds trivially
(both $\rr_\vphi$ and $\en$ vanish). At all other points,
$\|\d\vphi E_i\|^2=\frac2n\en>0$ independent of $i$ and $h(\d\vphi E_i,
\d\vphi E_j)=0$ for $i\neq j$, so 
$$
\hat{E_i}=\sqrt{\frac{n}{2\en}}\d\vphi E_i,\qquad i=1,2,\ldots,n
$$
form an orthonormal basis for $T_{\vphi(p)}N$. Hence,
for all $X,Y\in T_{\vphi(p)}N$,
\bea
h(X,\rr_\vphi Y)&=&\sum_{i=1}^nh(X,R^N(Y, \d\vphi E_i)\d\vphi E_i)
=\frac{2\en}{n}\sum_{i=1}^n h(X,R^N(Y, \hat{E}_i)\hat{E}_i)\nonumber \\
&=&
-\frac{2\en}{n}\sum_{i=1}^n h(\hat{E}_i,R^N(Y,\hat{E}_i)X) = 
\frac{2\en}{n}\ric(Y,X)
= \frac{2\en}{n}\frac{\kappa}{n} h(X,Y)
\eea
by a standard symmetry of $R^N$ \cite[p.\ 58]{wil} and 
the Einstein property of $(N,h)$. \end{proof}

\noindent Note that every $\vphi\in\rat_1$ is holomorphic, hence
conformal, and the unit two-sphere is Einstein with $\kappa=2$. Hence
\beq
\rr_\vphi= \en\, \id
\eeq
in this case. Thinking of $\jac_\vphi$ as a quantum Hamiltonian acting
on ``wavefunctions'' on the two-sphere (sections of $\vphi^{-1}TN$, really),
the effect of the curvature term is to add a potential well equal to
minus the classical energy density of the lump. What is this energy density?
By $SO(3)\times SO(3)$ invariance, the spectrum of $\jac_\vphi$
can depend only on $\lambda$, so it suffices to consider only the 
one-parameter family of rational maps 
\beq\label{aldi}
\vphi:z\mapsto W=\mu(\lambda)z=\frac{\Lambda +\lambda}{\Lambda-\lambda}z
\eeq
corresponding to the curve $(\I_3,(0,0,\lambda))$, $\lambda\geq 0$,
in $SO(3)\times \R^3$. It is convenient to parametrize these by
$\mu\in[1,\infty)$ rather than $\lambda\in[0,\infty)$. In terms
of the usual polar coordinates on both domain and codomain spheres,
the map (\ref{aldi}) is
\beq
\vphi:(\theta,\phi)\mapsto (f_\mu(\theta),\phi),\qquad\mbox{where}\quad
f_\mu(\theta)=2\cot^{-1}\left(\mu\cot\frac\theta2\right).
\eeq
The pair $E_1=\cd/\cd\theta$, $E_2=\cosec\theta\cd/\cd\phi$ is orthonormal
on $S^2$, whence, by conformality, one sees that
\beq
\en=\|\d\vphi E_2\|^2=\frac{\sin^2 f_\mu(\theta)}{\sin^2\theta}.
\eeq

We turn now to the explicit computation of the rough Laplacian
$\triangle_\vphi$, for $\vphi$ of the form (\ref{aldi}). Note that $\wt{E}_1=
E_1\circ\vphi$, $\wt{E}_2=E_2\circ\vphi$ gives an orthonormal pair
of sections of $\vphi^{-1}TN$. Any section $Y$ of $\vphi^{-1}TN$ can
be uniquely written $Y=Y_1(\theta,\phi)\wt{E}_1+Y_2(\theta,\phi)\wt{E}_2$.
We seek an expression for $\triangle_\vphi$ as a differential operator acting
on the pair of smooth functions $(Y_1,Y_2)$. In fact, we can deduce all
we need once we know how $\triangle_\vphi$ acts on sections of
the form $a(\theta)\cos m\phi \wt{E}_1$ where $m\in\N$. 
A straightforward but lengthy
computation, presented in the appendix, shows that
\bea\label{raaba}
\triangle_\vphi(a(\theta)\cos m\phi\wt{E}_1)&=&(\dd_m a)\cos m\phi\wt{E}_1
+Q_ma\sin m\phi\wt{E}_2\\
\mbox{where}\qquad
\dd_m&=&-\frac{d^2\:}{d\theta^2}-\cot\theta\frac{d\:}{d\theta}+
\frac{m^2+\cos^2f_\mu}{\sin^2\theta}\\
Q_m&=&2m\frac{\cos f_\mu}{\sin^2\theta}.
\eea
Since $N$ is k\"ahler, $\nabla^N$ commutes with $J^N$, the almost complex 
structure on $N$, namely $[\nabla^\vphi,J^N]=0,$ and hence 
$[\triangle_\vphi,J^N]=0$ also. So
we conclude immediately that
\bea
\triangle_\vphi(a(\theta)\cos m\phi \wt{E}_2)&=&
\triangle_\vphi(J^Na(\theta)\cos m\phi \wt{E}_1)=
J^N\triangle_\vphi(a(\theta)\cos m\phi \wt{E}_1)\nonumber \\ &=&
-Q_ma\sin m\phi\wt{E}_1+(\dd_m a)\cos m\phi\wt{E}_2.
\eea
Consider now the $SO(2)$ action on $C^\infty(S^2,S^2)$ given by
\beq
\vphi\mapsto\vphi_\alpha=R_\alpha\circ\vphi\circ R_{-\alpha}\qquad\mbox{where}
\qquad
R_\alpha=\left(\begin{array}{ccc}\cos\alpha&-\sin\alpha&0\\
\sin\alpha&\cos\alpha&0\\
0&0&1\end{array}\right),\quad\alpha\in\R.
\eeq
Note that $E(\vphi_\alpha)\equiv E(\vphi)$ and each of the maps (\ref{aldi})
is fixed under this $SO(2)$ action. It follows that $\jac_\vphi$, and hence,
by Proposition \ref{lidl}, $\triangle_\vphi$, commute with the induced action
on $\Gamma(\vphi^{-1}TN)$,
\beq
Y\mapsto \rot_\alpha(Y)=\d R_\alpha\circ Y\circ R_{-\alpha},
\eeq
or, in terms of polar coordinates,
\beq
\rot_\alpha(Y_1(\theta,\phi)\wt{E}_1+Y_2(\theta,\phi)\wt{E}_2)=
Y_1(\theta,\phi-\alpha)\wt{E}_1+Y_2(\theta,\phi-\alpha)\wt{E}_2.
\eeq
Hence
\bea
\triangle_\vphi(a(\theta)\sin m\phi\wt{E}_2)&=&
\triangle_\vphi(\rot_{\frac{\pi}{2m}}a(\theta)\cos m\phi \wt{E}_2)=
\rot_{\frac{\pi}{2m}}\triangle_\vphi(a(\theta)\cos m\phi
\wt{E}_2)\nonumber \\ &=& 
\rot_{\frac{\pi}{2m}}(-Q_ma\sin m\phi\wt{E}_1+(\dd_m a)\cos
m\phi\wt{E}_2)\nonumber \\ &=& 
Q_ma\cos m\phi\wt{E}_1+(\dd_m a)\sin m\phi\wt{E}_2,
\eea
and similarly
\beq
\triangle_\vphi(a(\theta)\sin m\phi\wt{E}_1)=(\dd_m a)\sin m\phi\wt{E}_1-
Q_m a\cos m\phi\wt{E}_2.
\eeq

To summarize, $\jac_\vphi$ leaves each infinite dimensional subspace
$I_m\subset\Gamma(\vphi^{-1}TN)$, $m\in\N$,
\beq
I_m=\{a_1(\theta)\cos m\phi\wt{E}_1+
a_2(\theta)\sin m\phi\wt{E}_2+
a_3(\theta)\cos m\phi\wt{E}_2+
a_4(\theta)\sin m\phi\wt{E}_1\: :\: a\in C^\infty((0,\pi),\R^4)\}
\eeq
invariant, these spaces span $\Gamma(\vphi^{-1}TN)$, and on $I_m$,
\beq
\jac_\vphi\left(\begin{array}{c}a_1\\a_2\\a_3\\a_4\end{array}\right)=
\left(\begin{array}{cccc}
\dd_m-\en& Q_m&0&0\\
Q_m&\dd_m -\en&0&0\\
0&0&\dd_m-\en&-Q_m\\
0&0&-Q_m&\dd_m-\en\end{array}\right)
\left(\begin{array}{c}a_1\\a_2\\a_3\\a_4\end{array}\right).
\eeq
This diagonalizes after a simple change of coordinates. Let
$\alpha_1=a_1+a_2$, 
$\alpha_2=a_3+a_4$, $\alpha_3=a_1-a_2$ and $\alpha_4=a_3-a_4$. Then
\beq
\jac_\vphi\left(
\begin{array}{c}\alpha_1\\\alpha_2\\\alpha_3\\\alpha_4\end{array}
\right)=
\left(\begin{array}{cccc}
\dd_m+Q_m-\en&0&0&0\\
0&\dd_m+Q_m-\en&0&0\\
0&0&\dd_m-Q_m-\en&0\\
0&0&0&\dd_m-Q_m-\en\end{array}\right)
\left(\begin{array}{c}\alpha_1\\\alpha_2\\\alpha_3\\\alpha_4
\end{array}\right). 
\eeq
Note that $\dd_{-m}=\dd_m$ and $Q_{-m}=-Q_m$, so $\spec\jac_\vphi$ is
the union of the spectra of the Sturm-Liouville operators
\beq\label{sl}
S_m=\dd_m+Q_m-\en=-\frac{d^2\:}{d\theta^2}-\cot\theta\frac{d\:}{d\theta}
+U_m(\theta),\qquad m\in\Z,
\eeq
where
\beq
U_m(\theta)=\frac{m^2-1+2\cos^2f_\mu(\theta)+
2m\cos f_\mu(\theta)}{\sin^2\theta}
\eeq
and each eigenvalue occurs with double multiplicity. 
In terms of physics, we may think of
this as the energy spectrum for ``spin 0'' (i.e.\ $\phi$ independent)
states of a point particle moving on $S^2$ in the $SO(2)$ invariant
potential well $U_m(\theta)$. The spectrum of each $S_m$ may be computed
numerically using the shooting method described in \cite{hasspe}.

There is one value of $\mu$ for which $\spec\jac_\vphi$ may be
computed exactly, namely $\mu=1$. Here things simplify considerably, 
because the corresponding rational map is the identity map $S^2\ra S^2$. 
The Jacobi operator for the identity map on a general Riemannian manifold 
$(M^n,g)$ was studied in detail by Smith \cite{smi}. The key 
simplification is that one has a canonical identification 
$\id^{-1}TN\equiv T^*M$, obtained
by identifying the section $Y$ of $\id^{-1}TN=TM$ with the one form
$\flat\, Y=g(Y,\cdot)$. This is useful because there is a Weitzenb\"ock
formula relating the rough Laplacian $\triangle_\id$ to the Hodge Laplacian 
$\triangle$ on
one-forms \cite[p.\ 161]{ura}
\beq
\triangle_\id Y=\sharp\, (\triangle\, \flat\, Y-\ric(Y,\cdot)).
\eeq
In the case where $M^n$ is Einstein, with (constant) scalar curvature 
$\kappa$, 
\beq
\sharp\, \ric(Y,\cdot)\equiv\frac{\kappa}{n}Y,
\eeq
which together with Proposition \ref{lidl} and the observation that
$\en=\frac{n}{2}$, constant, for $\id$, gives the following formula
for $\jac_\id$ on an Einstein manifold, originally due to Smith,
\beq
\jac_\id=\sharp\, \triangle\, \flat-\frac{2\kappa}{n}.
\eeq
In our case, $n=\kappa=2$, so $\jac_{\id}=\sharp\, \triangle\, \flat\, -2$, 
that is,
the spectrum of $\jac_\vphi$ in the special case $\mu=1$ is the spectrum of
the Hodge Laplacian on one-forms on $S^2$, shifted down by $2$. This, in
turn, can be related to the spectrum of the Laplacian on functions on $S^2$,
as follows.

Let $\Omega_p$ denote the space of smooth $p$-forms on $S^2$, and
$\dstar:\Omega_p\ra\Omega_{p-1}$ denote the coderivative, so that the
Hodge Laplacian is $\triangle=\dstar\d+\d\dstar$. Every one-form on $S^2$
has a unique Hodge decomposition
\beq
Y_1=\d Y_0+\dstar Y_2
\eeq
into exact and coexact components (the harmonic component vanishes as 
$H^1(S^2)=0$). Now $\d\Omega_0$ and $\dstar\Omega_2$ are $L^2$ orthogonal
subspaces of $\Omega_1$ and $[\triangle,\d]=[\triangle,\dstar]=0$, so the
spectral problem for $\triangle_{\Omega_1}$ decomposes into two
sub-problems, for  
$\triangle|_{\d\Omega_0}$ and $\triangle|_{\dstar\Omega_2}$. Since
$Y_1$ is coexact if and only if $*Y_1$ is exact, and
$*:\Omega_1\ra\Omega_1$ is a $L^2$ isometry, the spectra of
$\triangle|_{\d\Omega_0}$ and $\triangle|_{\dstar\Omega_2}$ 
coincide. Hence $\spec\triangle_{\Omega_1}$ is
$\spec\triangle|_{\d\Omega_0}$ with double multiplicity. In fact, 
\beq
\spec\triangle|_{\d\Omega_0}=\spec\triangle_{\Omega_0}\less\{0\}.
\eeq
To see this, let $\nu\in\spec\triangle_{\Omega_0}$ and $Y_0$ be the
corresponding eigenfunction. Then
\beq
\triangle\d Y_0=\d\triangle Y_0=\d\nu Y_0=\nu\d Y_0
\eeq
so if $\d Y_0\neq 0$ (that is, $\nu\neq 0$), then 
$\nu\in\spec\triangle|_{\d\Omega_0}$.
Conversely, let $\nu\in\spec\triangle|_{\d\Omega_0}$ and 
$\d Y_0$ be the corresponding eigenform. Then $\nu\neq 0$ (since there is no
harmonic one-form on $S^2$), so
\bea
0=\triangle\d Y_0-\nu\d Y_0&=&\d(\triangle Y_0-\nu Y_0)\nonumber \\
\Rightarrow\quad\triangle Y_0-\nu Y_0&=&c,\quad\mbox{constant}\nonumber \\
\Rightarrow\quad\triangle Y_0'-\nu Y_0'&=&0,\qquad 
\mbox{where $Y_0'=Y_0+\frac{c}{\nu}$}
\eea
and hence $\nu\in\spec\triangle_{\Omega_0}$. The spectrum of
$\triangle_{\Omega_0}$ is well known. 

We conclude that, at $\mu=1$ (equivalently $\lambda=0$),
\bea
\spec\jac_\vphi&=&\{\ell(\ell+1)-2\: :\: \ell\in\Z^+\}\nonumber \\
\mbox{multiplicity}(\ell(\ell+1)-2)&=&4\ell+2.
\eea
Since we know the eigenvalues (and eigensections) of $\jac_\vphi$ at
$\lambda=0$, we can use these as seed data for the numerical shooting method,
starting at $\lambda=0$ and increasing $\lambda$ in small steps. In this
way we can numerically construct curves $\omega_i^2(\lambda)$ showing how the
different eigenvalues vary with lump sharpness $\lambda$. Figure \ref{jacspec}
shows such curves for the lowest 48 eigenvalues. Recall that the 
eigenvalue $0$ has multiplicity 6 (the dimension of the moduli space), and
that other eigenvalues always have multiplicity (at least) 2 due to
the symmetry under $J^N$ (or, equivalently, due to the block structure
of $\jac_\vphi$ on $I_m$). 

\begin{figure}[!ht]
\begin{center}
\includegraphics[scale=0.65]{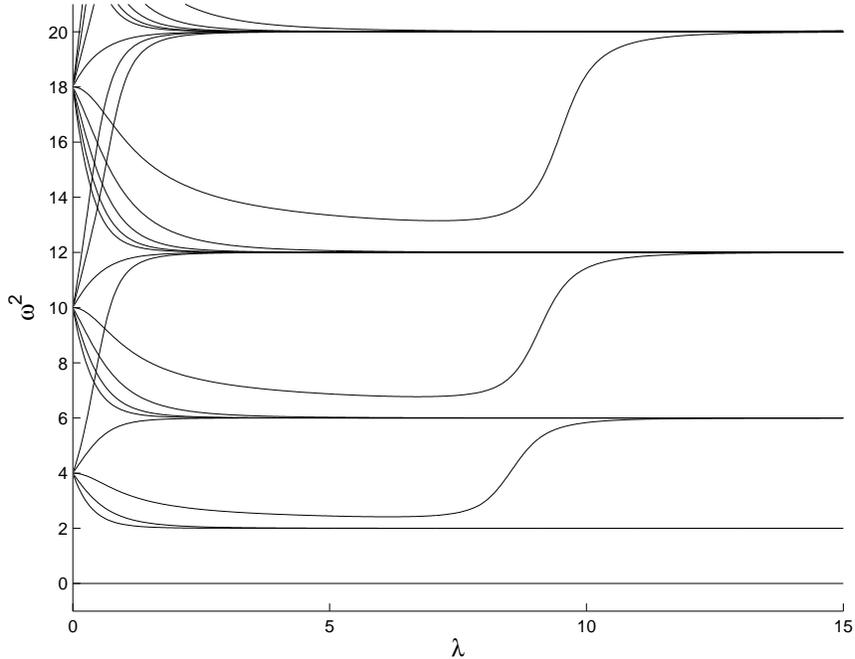}
\end{center}
\caption{The dependence of the eigenvalues of the Jacobi operator
$\jac_\vphi$ on the lump sharpness $\lambda$. Note that the eigenvalues 
interpolate between $\ell(\ell+1)-2$, $\ell\in\Z^+$, at $\lambda=0$
and $\ell(\ell+1)$, $\ell\in\N$ as $\lambda\ra\infty$.}
\label{jacspec}
\end{figure}

It is interesting to examine the $\lambda\ra\infty$ behaviour of 
$\omega_i^2(\lambda)$. The pointwise limit of $\vphi:S^2\ra S^2$ as
$\lambda\ra\infty$ is
\beq
\vphi_\infty(p)=\left\{\begin{array}{ll}
(0,0,1)&\mbox{$p\neq(0,0,-1)$}\\
(0,0,-1)&\mbox{$p=(0,0,-1)$}
\end{array}\right.
\eeq
that is, $\vphi_\infty$ is constant almost everywhere. A sensible guess
for the limiting spectrum would, therefore, be the spectrum of the
Jacobi operator at a constant map, which is known to coincide with $n$
copies of the
spectrum of the Laplacian on functions, where $n$ is the dimension of the
codomain \cite[p160]{ura}. In this case
\bea
\spec\jac_{\rm const}&=&\{\ell(\ell+1)\: :\: \ell\in\N\}\nonumber \\
\mbox{multiplicity}(\ell(\ell+1))&=&4\ell+2.
\eea
The numerics suggest that this guess for the limiting spectrum
is very nearly correct. Specifically, the eigenvalues do tend, as $\lambda
\ra\infty$ to eigenvalues of $\jac_{\rm const}$, and, apart from the 
eigenvalues $0$ and $2$, those eigenvalues tending to $\ell(\ell+1)$
have total multiplicity $4\ell+2$. So the guess is wrong only in that it 
predicts the multiplicity of the limiting eigenvalue $0$ to be
$2$ rather than $6$ (as it must be given the dimension of $\rat_1$)
and the multiplicity of $2$ to be $6$ rather than $4$ (as found 
numerically).

\subsection{The regularized Casimir energy}

Having computed $\spec\jac_\vphi$, we must now try to make sense of the
Casimir energy $\Cas(\lambda)$ in 
(\ref{Cformal}) which, as
it stands, is divergent. It is conventional to reset the zero of
potential energy so that the total zero-point energy of the vacuum is $0$.
For our purposes, $\vphi_\infty$ is (almost everywhere) the vacuum, so
it is convenient to define
\beq
\label{Cask}
\Cas_k(\lambda)=\frac12\sum_{i=1}^k(\omega_i(\lambda)-\omega_i(\infty)),
\eeq
where the eigenvalues $\omega_i(\lambda)^2>0$ are arranged 
in such a way that $\omega_i(0)^2$ is nondecreasing.

One option would be to renormalize the spectrum (\ref{Cask}) numerically, 
using for example the heat kernel approach. Here, we will employ a 
different semi-analytical approach. We first discuss the finite sums 
$\Cas_k$ for special values of $k$ and then describe how we regularize 
the diverging sum as $k \to \infty.$ We choose $k=10$, $k=24$ and $k=42$, 
which include the lowest $1$, $2$ and $3$ eigenvalues of $\jac_\id$ 
respectively.
Note that, since we have eigenvalue crossings in figure \ref{jacspec}, this 
amounts to the lowest $k$ normal modes at $\lambda=0$, but not at large 
$\lambda$ (where the eigenvalue ordering has changed). If we were to define 
$\Cas_k(\lambda)$ as the sum of the frequencies of the lowest $k$ normal 
modes at each $\lambda$, the function $\Cas_k$ would not be smooth (it would 
have corners where eigenvalues cross). 
In effect, we are making a large but finite-dimensional approximation to
the quantum field theory, in which the wavefunction is a function on a 
vector bundle over $\rat_1$, whose fibre over $\vphi$ is a union of
low-lying eigenspaces of $\jac_\vphi$. We are choosing these eigenspaces so
that they vary smoothly over $\rat_1$. The price for this is that they are,
towards the boundary of $\rat_1$, not quite the lowest energy eigenspaces
available up to dimension $k$.

\begin{figure}[ht!]
\begin{center}
\includegraphics[scale=0.65]{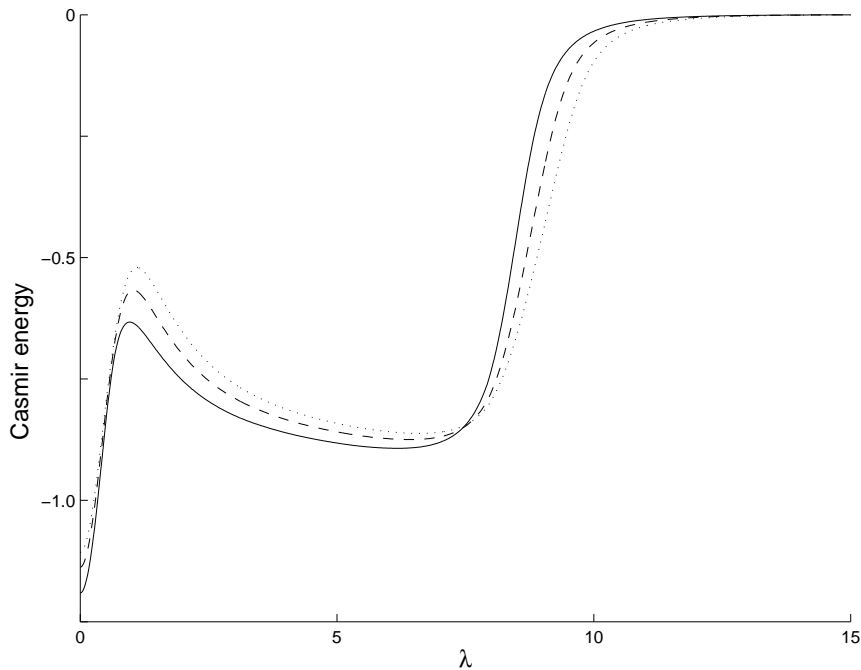}
\end{center}
\caption{The one-lump Casimir energy as a function of lump sharpness $\lambda$
for three different cut-offs: solid curve $\Cas_{10}(\lambda)$, dashed curve
$\frac12\Cas_{24}(\lambda)$, dotted curve
$\frac13\Cas_{42}(\lambda),$ where $\Cas_k$ is defined in
(\ref{Cask}). The values $k=10,$ $k=24$ and $k=42$ correspond to the
lowest 1, 2, and 3 eivenvalues of ${\cal J}_{{\rm Id}}$, respectively.}
\label{casplot}
\end{figure}

Plots of $\Cas_{10}$, $\Cas_{24}$ and $\Cas_{42}$ 
 are presented in figure \ref{casplot}. 
 The curves for $\Cas_{24}$  
and $\Cas_{42}$ have been rescaled vertically, by a factor of $\frac12$ 
and $\frac13$ respectively, to make comparison with $\Cas_{10}$ easier.
Note that these three functions are qualitatively very similar. In fact, 
the Casimir curves $\Cas_{10},$ $\Cas_{24}$ and $\Cas_{42}$ are
approximately self-similar up to a factor which diverges as $k \to \infty.$ 
We have opted to use the self-similarity and regularize the diverging 
factor (the depth of the well at $\lambda=0$). So we define
the approximate renormalized Casimir energy $\Cas(\lambda)$ to be
\beq
\label{Cas}
\Cas(\lambda)=C_*\frac{\Cas_{10}(\lambda)}{|\Cas_{10}(0)|}
\eeq
where $C_*$ is the renormalized Casimir energy of the $\lambda=0$
lump. This can be computed exactly using zeta
function regularization, because, as we have seen
 the spectra of the Jacobi operator for
the identity map ($\lambda=0$ lump) and the constant map (the
vacuum) are known exactly.

\subsection{Zeta function regularization}

For $\lambda = 0,$ the spectrum of the Jacobi operator is known
explicitly, namely $\omega_0^2 = 0$ with
multiplicity $\mu_0 = 6$ 
and $\omega_l^2 = l(l+3)$ with multiplicity $\mu_l = 4l+6$
for $l=1,2,\dots$ This enables us to calculate the Casimir energy
using zeta function regularization. The key idea is to write down the
corresponding zeta function 
\begin{equation}
\label{zeta}
\zeta(\nu) = \sum\limits_{l=1}^\infty \mu_l  \left(\omega_l^2
\right)^{-\nu},
\end{equation}
leaving out the zero modes. The zeta function (\ref{zeta}) is
absolutely convergent as long as the real part of $\nu$ is sufficiently
large. In this region, $\zeta(\nu)$ can be viewed as an analytic
function of $\nu$. We are interested in the Casimir energy which
corresponds to the value $\nu=-\frac12$. In this case, the formal sum in
(\ref{zeta}) is divergent, however, $\zeta(\nu)$ defined as analytic
continuation is well-defined. It is convenient to rewrite $\zeta(\nu)$ as
the following sum:
\beq
\zeta(\nu) = \sum\limits_{l=1}^\infty 2 l (l(l+3))^{-\nu} + 
\sum\limits_{l=1}^\infty 2 (l+3) (l(l+3))^{-\nu} 
\eeq
then we can use formula (5.8) in \cite[p.\ 122]{eli} to obtain
$$
\zeta(-\tfrac12) = -2.373.
%\zeta(-1/2) = -2.372868064.
$$
The corresponding calculation for the vacuum, whose  spectrum is 
$\omega_l^2 =l(l+1)$ 
with multiplicity $\mu_l = 4l+2$, leads to the zeta function
\beq
\zeta(\nu) = \sum\limits_{l=1}^\infty 2 l (l(l+1))^{-\nu} + 
\sum\limits_{l=1}^\infty 2 (l+1) (l(l+1))^{-\nu} 
\eeq
and yields 
$$
\zeta_{{\rm vac}}(-\tfrac12) = -0.530.
%\zeta_{{\rm vac}}(-1/2) = -0.5301910982.
$$
See also equation (5.34) in \cite[p. 126]{eli}.
Hence, the Casimir energy of the $\lambda=0$ lump on the unit
two-sphere can be evaluated using
$$
\frac12 \left(\zeta(-\tfrac12) - \zeta_{{\rm vac}}(-\tfrac12) \right)
= -0.921.
%= -0.9213384829.
$$
Hence
$$
C_*=|\Cas(0)| = 0.921.
$$
This is the energy scale we used to set the renormalized energy scale
for our numerically generated Casimir energy function 
$\Cas(\lambda)$ in (\ref{Cas}). 

Our approximation to the Casimir energy $\Cas(\lambda)$ is now 
given by the rescaled and shifted curve $\Cas_{10}(\lambda)$ in figure
\ref{casplot}. It is non-singular and appears to be smooth. It is worth
summarizing the approximations involved. Our calculation relies on a
conjectured self-similarity of the curves in figure \ref{casplot}. We
also assume that we can neglect the effects of crossing modes. In
particular, we assume that the zeta-function regularization can be
performed pointwise when we regularize the Casimir energy at
$\lambda=0$ and $\lambda=\infty.$ Furthermore, we assume that the
$\Cas(\lambda)\ra 0$ as $\lambda\ra 0$ because the lumps converge 
almost everywhere to the vacuum in that limit.

The obvious alternative to our approach, namely a
numerical evaluation using heat kernel or zeta function regularization
as in \cite{mosshitor, mos}, is beyond the scope of this paper.

\section{The energy spectra}
\news

In the following, we describe how to calculate the spectrum of
the Laplacian (\ref{lap}). In order to make use of the physics literature 
on this topic we rewrite the Laplacian using angular momentum 
operators. Note that the operators ${\bf J} = - i \thetavec$,
${\bf L} = -i \lamvec\times\cdv$, ${\bf S}={\bf J}+{\bf L}$
and ${\bf K}=-i\xivec$ all satisfy the canonical commutation
relations for angular momenta, namely, 
$$
[G_1, G_2] = i G_3,\qquad {\bf G}={\bf J}, {\bf K}, {\bf L}, {\bf S}
$$
and cyclic permutations. Recall that $i{\bf K}$ generates the left $SO(3)$
action, on the target $S^2$, and $i{\bf S}$ generates the right $SO(3)$
action, on the physical $S^2$, so we refer to these operators as
isospin and spin respectively. Making use of these operators, one sees that
 the Laplacian (\ref{lap}) is
\bea
\triangle \psi&=&
-\frac{1}{\Lambda\lambda^2A^2B}\frac{\cd\:}{\cd\lambda}\left(
\frac{\lambda^2\Lambda^3A^2}{4}\frac{\cd \psi}{\cd\lambda}\right)+
\frac{1}{\lambda^2\Lambda^2A}{\bf L}^2 \psi \nonumber \\
&&+ \frac{2}{\Lambda^2A}\left\{ {\bf J}^2 \psi+ {\bf S}^2\psi-
\frac{2}{\lambda^2}\left[1-
\frac{\Lambda^2A}{4B}\right](\lamvec \cdot {\bf S})^2\psi\right\}.
\label{lapnew}
\eea
Here, we replaced the ${\bf L} \cdot {\bf J}$ term using the
convenient operator identity
\beq
2 {\bf L}\cdot {\bf J} = {\bf S}^2 - {\bf L}^2 - {\bf J}^2.
\eeq
We have already shown that $\triangle$ commutes with ${\bf S}$ and ${\bf K}$,
as it must, by $SO(3)\times SO(3)$ invariance. In fact, the operators
$\triangle, |{\bf S}|^2, S_3, |{\bf K}|^2, K_3, P$,
where $P: {\pmb \lambda} \mapsto -{\pmb \lambda}$ is the parity operator,
 all mutually commute, so we can seek simultaneous eigenstates of these
six operators. 

\subsection{Spin-isospin interchange symmetry}
\label{interchange}

Let $s(s+1)$ and $k(k+1)$ be the eigenvalues of
$|{\bf S}|^2$ and $|{\bf K}|^2$. It is a somewhat
surprising fact that the spectrum of $\triangle$ is invariant under
interchange of $s$ and $k$. 
Even more surprising is that this remains
true even if we give the physical and target spheres different radii. 
The key to seeing this is the following isometry of any
$SO(3)\times SO(3)$ invariant k\"ahler metric on $\rat_1$.

\begin{prop} Identify $\rat_1\cong SO(3)\times\R^3$ with $TSO(3)$ via
  $(R,\lamvec)\equiv (\lamvec\cdot\thetavec)(R)\in T_RSO(3)$. The mapping
$f:SO(3)\ra SO(3)$, $f(R)=R^{-1}$ induces a mapping $df:TSO(3)\ra TSO(3)$.
Then $df$ is an isometry of any $SO(3)\times SO(3)$ invariant k\"ahler
metric on $TSO(3)$.
\end{prop}

\begin{proof} The identification $\rat_1\equiv TSO(3)$ amounts to
thinking of $\lamvec$ as a vector in the Lie algebra $\so(3)$ (skew
$3\times 3$ real matrices) and identifying $T_RSO(3)$ with $\so(3)=T_e SO(3)$
by left translation. Hence, the map in question is
$$
df:(R,\lamvec)\mapsto (R^{-1}, -Ad_R\lamvec).
$$
Let $q(t)=(R(t),\lamvec(t))$ be a curve in $\rat_1$ with 
$\dot{q}(0)=(R(0)\Omegavec,\vvec)$,
$\Omegavec,\vvec\in \so(3)$. Then, with respect to the $SO(3)\times SO(3)$
invariant k\"ahler metric determined by the functions $A_i(\lambda),$ as 
in Proposition \ref{metprop},
$$
\|\dot{q}(0)\|^2=A_1|\vvec|^2+A_2(\lamvec\cdot\vvec)^2+
A_3|\Omegavec|^2+A_4(\lamvec\cdot\Omegavec)^2+
A_1\vvec\cdot (\lamvec\times\Omegavec).
$$
The image of this curve under $df$ is $\wt{q}(t)=(R^{-1}(t),-R(t)\lamvec(t)
R^{-1}(t))$, which has 
$\dot{\wt{q}}(0)=(R^{-1}(0)\wt{\Omegavec},\wt{\vvec})$, where
$$
(\wt{\Omegavec},\wt{\vvec})= 
-(Ad_R\Omegavec,Ad_R(\vvec-[\lamvec,\Omegavec]_{\so(3)})).
$$
Recalling that the Lie bracket on $\so(3)$ coincides with minus the vector
product on $\R^3$ under the natural identification $\so(3)\equiv \R^3$, 
and the adjoint action of $SO(3)$ on $\so(3)$ coincides with the 
fundamental action on $\R^3$, we
find that
\ben
\|\dot{\wt{q}}(0)\|^2&=&A_1|\vvec+\lamvec\times\Omegavec|^2+
A_2(\lamvec\cdot\vvec)^2+
A_3|\Omegavec|^2+A_4(\lamvec\cdot\Omegavec)^2-
A_1(\vvec+\lamvec\times\Omegavec)\cdot(\lamvec\times\Omegavec).
\\
&=&
\|\dot{q}(0)\|^2.
\een
Hence $df$ is an isometry, as claimed.
\end{proof}

\begin{prop}\label{symmetry} Given $E\in\R$ and $k,s\in\N$, denote by 
$X_{E,k,s}$ the simultaneous eigenspace of
$$
H=\tfrac12\triangle+V,\qquad
|{\bf K}|^2,\qquad |{\bf S}|^2
$$
with eigenvalues $E$, $k(k+1)$ and $s(s+1)$, where $V:\rat_1\ra \R$
is any $SO(3)\times SO(3)$ invariant potential. Then $X_{E,k,s}$ and
$X_{E,s,k}$ have equal dimension.
\end{prop}

\begin{proof}
Under the identification $\rat_1\equiv TSO(3)$, the left 
and right 
actions of
$SO(3)$ on $\rat_1$ coincide with the natural
left and right actions of $SO(3)$ on $TSO(3)$. Denote by $df^*$ the
induced map on $L^2(TSO(3),\C)$, $df^*\psi=\psi\circ df$. Then, since
$df$ is an isometry and preserves the length of $\lamvec$,
$H df^*=df^*H$. Furthermore, $df$ interchanges
the left and right $SO(3)$ actions on $TSO(3)$, so ${\bf K}df^*=-df^*{\bf 
S}$ and ${\bf S}df^*=-df^*{\bf K}$.
Hence, $df^*:X_{E,k,s}\ra X_{E,s,k}$ and $df^*:X_{E,s,k}\ra X_{E,k,s}$
for all $E,k,s$.
Now $df^2=\id$, so $(df^*)^2=\id$, and hence $df^*$ is
invertible, so the eigenspaces $X_{E,k,s}$ and $X_{E,s,k}$ are isomorphic.
\end{proof}

Note  that this $k,s$ interchange symmetry holds
for {\em any} invariant k\"ahler metric on $\rat_1$, so is not special
to the case where the domain and target two-spheres have equal radius.
Also note that it relies heavily on the k\"ahler property of the metric,
and does not follow from $SO(3)\times SO(3)$ invariance alone (or, indeed,
just invariance and hermiticity). This is a ``hidden'' symmetry
which one cannot see directly from the field theory. It is an
interesting question whether a direct physical argument can be given to
explain this symmetry. A construction involving a supersymmetric 
extension of the model would be a natural candidate.

\subsection{Reduction to a Sturm-Liouville problem}

The wavefunction is a function $\psi:SO(3)\times\R^3\ra\C$.
According to the Peter-Weyl theorem \cite{petwey}, the matrix elements of
all irreducible unitary representations of the group $SO(3)$ form
an orthonormal basis for $L^2(SO(3),\C)$. Recall that such 
representations are labelled by $k\in\{0,1,2,\ldots\}$, and that the
matrix elements of the $k$ representation are functions
$\pi^{(k)}_{j_3,k_3}:SO(3)\ra\C$ where $-k\leq j_3,k_3\leq k$. We have 
chosen the symbol $k$ to label the representations because
\beq
|{\bf K}|^2\pi^{(k)}_{j_3,k_3}=
|{\bf J}|^2\pi^{(k)}_{j_3,k_3}=k(k+1)\pi^{(k)}_{j_3,k_3}.
\eeq
Further, we may choose the basis for $\C^{2k+1}$ so that
\beq
K_3\pi^{(k)}_{j_3,k_3}=k_3\pi^{(k)}_{j_3,k_3},\qquad
J_3\pi^{(k)}_{j_3,k_3}=j_3\pi^{(k)}_{j_3,k_3},
\eeq
which is why we have labelled the matrix elements with $j_3,k_3$. 
For the $\R^3$ dependence of $\psi$ we use spherical polar coordinates and
expand the angular dependence in spherical harmonics.
Hence
we may express the wavefunction as follows,
\beq
\psi=\sum_{k=0}^\infty\sum_{l=0}^\infty\sum_{j_3=-k}^k\sum_{k_3=-k}^k
\sum_{l_3=-l}^l a^{k,l}_{j_3,k_3,l_3}(\lambda)\ket{k,l,j_3,k_3,l_3},
\eeq
where
\beq
\ket{k,l,j_3,k_3,l_3}=\pi^{(k)}_{j_3,k_3}Y_{l,l_3},
\eeq
and $Y_{l,l_3}:S^2\ra\C$ denote spherical harmonics. 

Recall that
 $H,|{\bf K}|^2,K_3,|{\bf S}|^2,S_3$ are mutually commuting, so we may
solve the
eigenvalue problem for $H$ on each simultaneous
eigenspace of $|{\bf K}|^2,K_3,|{\bf S}|^2,S_3$
separately. Clearly,
$H$ is independent of $k_3$ and $s_3$ so we may, and henceforth will, 
without loss
of generality, set $k_3=s_3=0$, remembering to multiply all degeneracies
by $(2k+1)(2s+1)$ to account for the other values of $(k_3,s_3)$. Furthermore,
by Proposition \ref{symmetry} we may, without loss of generality,
assume that $k\leq s$, doubling the multiplicity if $k<s$.
Recall that ${\bf S}={\bf J}+{\bf L}$ and that ${\bf J}$ and ${\bf L}$
satisfy the angular momentum algebra. Hence a
basis $\{\ket{k,s,l}\: :\: |k-s|\leq l\leq k+s\}$
for the $(k,0,s,0)$ eigenspace can be constructed, using Clebsch-Gordon
coefficients, on which the operators $|{\bf J}|^2=
|{\bf K}|^2$, $|{\bf S}|^2,$ $|{\bf L}|^2$ 
act naturally
\bea
|{\bf J}|^2 \ket{k,s,l} &=& k(k+1) \ket{k,s,l}\\
|{\bf S}|^2 \ket{k,s,l} &=& s(s+1) \ket{k,s,l}\\
|{\bf L}|^2 \ket{k,s,l} &=& l(l+1) \ket{k,s,l},
\eea
see Appendix \ref{AngularMomentum} for details.

By expanding the wavefunction as 
$\psi=\sum_{l=|k-s|}^{k+s}a_l(\lambda)\ket{k,s,l}$,
we may express the Hamiltonian 
$H=\frac12\triangle +V$ as a $m\times m$ matrix of differential
operators acting on the vector function 
$a:[0,\infty)\ra \C^m$, where $m=2\min(k,s)+1$. Its structure is
\beq
Ha=-p_1(\lambda)\frac{d\:}{d\lambda}\left(p_2(\lambda)\frac{da}{d\lambda}
\right)+p_3(\lambda)M_1a+p_4(\lambda)(k(k+1)+t(t+1))a+p_5(\lambda)M_2a
+V(\lambda)a
\eeq
where $p_i(\lambda)$ are rather complicated but explicitly known functions 
of $\lambda$ 
(see Proposition \ref{lapprop}) and $M_1$, $M_2$ are $m\times m$ matrices
corresponding to ${\bf L}^2$ and $(\lamvec\cdot {\bf S})^2$ respectively.
By our choice of basis, $M_1$ is diagonal and has entries $l(l+1)$ with
$l$ running from $|k-s|$ (top left) to $k+s$ (bottom right). 
The matrix $M_2$ corresponding to operator $(\lamvec\cdot{\bf S})^2$ is
the only non-diagonal term in the Hamiltonian and is
 discussed in more detail in 
appendix \ref{AngularMomentum}. It mixes states of different 
$l$, which differ by 2.\, This leads to a natural chess board structure. 
Hence, by a reordering of the basis vectors, the operator can be written in 
block diagonal form with one block corresponding to the states with
even $l$ and the other block to odd $l$. These blocks correspond to the
decompostion of the $(k,0,s,0)$ eigenspace into $P=+1$ (even $l$) and
$P=-1$ (odd $l$) parity eigenspaces noting that $P Y_{l,m} = (-1)^l 
Y_{l,m}.$
In summary, for fixed $(k,s)$, the eigenvalue problem for any Hamiltonian
of the form $H=\frac12\triangle+V(\lambda)$ reduces to
a matrix-valued Sturm-Liouville problem on $[0,\infty)$ of dimension
$m=2\min(k,s)+1$.

\subsection{Boundary conditions}

We can now address the eigenvalue equations
\begin{equation}
\label{Heq}
H_i a = E a,
\end{equation}
where $E$ is the energy eigenvalue, and the Hamiltonians $H_0,$ $H_1$
and $H_2$ are given by  (\ref{Hamiltonians}). 
As just explained, the spectral problem reduces to a sequence of 
matrix-valued Sturm-Liouville problems indexed by $s\in\{0,1,2,\ldots\}$
and $k\in\{0,1,\ldots,s\}$, where each subproblem has 
dimension $2\min(k,s)+1.$ 
In order to calculate the spectrum of the Hamiltonians $H_0,$ $H_1$ and 
$H_2$ we have to derive not only the relevant differential equations but 
also the appropriate boundary conditions. For many important examples in 
mathematical physics, the boundary conditions
are determined solely by the requirement that the wavefunction be
$L^2$ finite. 
However, in our case, $L^2$ finiteness is not always sufficient, and 
 we have to apply the theory of singular Sturm-Liouville 
equations, following \cite{Pryce}. 

The asymptotic behaviour of the $L^2$ metric on $\rat_1$ can be
obtained by direct calculation:
\beq
\begin{array}{llll}
A(\lambda) \sim \frac{32\pi}{3}\quad& {\rm as~} \lambda \to 0, \quad&
A(\lambda) \sim \frac{8\pi}{\lambda^2}& {\rm as~} \lambda \to \infty,\\
B(\lambda) \sim \frac{8\pi}{3}& {\rm as~} \lambda \to 0, &
B(\lambda) \sim \frac{4\pi \log \lambda}{\lambda^4}
& {\rm as~} \lambda \to \infty.
\end{array}
\eeq
Using the above limits, we obtain the leading order equation for $(H_0-E)f=0$ 
as $\lambda \to \infty,$
\beq
\label{indiciali}
\frac{\partial^2 f}{\partial \lambda^2} +
\frac{1}{\lambda}\frac{\partial f}{\partial \lambda} 
- \frac{4}{\lambda^2}\left({\hat \lamvec}\cdot {\bf S}\right)^2 f = 0,
\quad {\rm with} \quad {\hat \lamvec} = \frac{\lamvec}{\lambda}, 
\eeq
and the leading order equation for $(H_0-E)f=0$ as $\lambda \to 0,$
\beq
\label{indicial0}
\frac{\partial^2 f}{\partial \lambda^2} +
\frac{2}{\lambda}\frac{\partial f}{\partial \lambda} 
- \frac{l(l+1)}{\lambda^2} f = 0.
\eeq
For $\lambda \to 0$ equation (\ref{Heq}) for $H_0$ has a regular
singular point. The asymptotic form of the solutions can be
derived from (\ref{indicial0}) and is given by\footnote{Higher order
  terms in the expansions for $\lambda \to 0$ and $\lambda \to \infty$
  have been calculated in \cite{McGlade} for $H_0$ with $k=0.$} 
\beq
f(\lambda) = c_1 \lambda^l 
+ c_2 \lambda^{-l-1}.
\eeq
Note that the leading order asymptotic behaviour is independent of the
energy eigenvalue $E$.
We are interested in solutions that are $L^2$ finite which leads to
the condition
\beq
\label{L2finite}
\int f^2 {\rm vol} < \infty.
\eeq
Using polar coordinates for (\ref{vol}) we obtain the two asymptotic
behaviours, 
$$
f^2 {\rm vol} \sim  {\tilde c}_1 \lambda^{2l+2} \quad {\rm and} \quad
f^2 {\rm vol} \sim  {\tilde c}_2 \lambda^{-2l},
$$
for suitable constants ${\tilde c}_1$ and ${\tilde c}_2.$
Hence $l>0$ corresponds to the limit-point case, since only the first
solution is $L^2$ finite, see e.g. \cite{Pryce} for further details.
However, for $l=0$ the situation is slightly
more subtle. Both asymptotic solutions are $L^2$ finite. This is known
as the limit-circle case.\footnote{The same subtlety occurs when
solving the Schr\"odinger equation for the hydrogen atom in
spherical polar coordinates.} 
The Hamiltonians $H_1,H_2$ lead to the same asymptotic behaviour since the
curvature  and Casimir energy
are finite as $\lambda\to 0$.
 
For $\lambda \to \infty$ it is convenient to analyse the boundary
conditions in a different basis so that the operator 
$({\hat \lamvec}\cdot {\bf S})$ is diagonal,
$$
\left({\hat \lamvec}\cdot {\bf S}\right)^2 f = p^2 f,
$$
where the integer $p$ satisfies $-\min(k,s) \le p \le  \min(k,s),$ see
appendix \ref{AngularMomentum}. 
Then the asymptotic solution follows from (\ref{indiciali}) and is given by
\beq
\label{asymk}
f(\lambda) = c_1 \lambda^{-2p} 
+ c_2 \lambda^{2p},
\eeq
for $p \neq 0$ and
\beq
f(\lambda) = c_1 + c_2 \log(\lambda),
\eeq
for $p=0.$ Again the leading order behaviour is independent of the
energy eigenvalue $E.$

Both $H_1$ and $H_2$ contain the scalar curvature function $\kappa(\lambda)$
which is known to diverge to infinity as $\lambda\ra\infty$. In fact,
using the formula obtained in \cite{spe3}, and accounting for the change
in normalizations (recall we are giving the domain and target spheres
unit radius), we find the asymptotic formula
\beq
\kappa(\lambda)\sim\frac{1}{16\pi}\frac{\lambda^4}{\log\lambda}\qquad
\mbox{as $\lambda\ra\infty$}.
\eeq
This adds a term to (\ref{indiciali}), namely,
\beq
\label{indicialik}
\frac{\partial^2 f}{\partial \lambda^2} +
\frac{1}{\lambda}\frac{\partial f}{\partial \lambda} 
- \frac{1}{\lambda^2}\left(4\left({\hat \lamvec}\cdot {\bf S}\right)^2
+\frac{1}{2\log(\lambda)^2} 
\right) f = 0.
\eeq
For $p\neq 0$ we can solve (\ref{indicialik}) in terms of modified
Bessel functions 
\beq
f(\lambda) = 
c_1 \sqrt{\log \lambda} K_{\sqrt{3}/2}\left(2 p \log (\lambda)\right)
+ c_2 \sqrt{\log \lambda} I_{\sqrt{3}/2}\left(2 p \log (\lambda)\right).
\eeq
With the help of asymptotic expansions, see e.g. \cite{abra}, it
can be shown that the leading order is identical to (\ref{asymk}).
For $p=0$ we obtain
\beq
f(\lambda) = c_1 \log(\lambda)^{\frac{1}{2} - \frac{1}{2} \sqrt{3}} 
+ c_2 \log(\lambda)^{\frac{1}{2} + \frac{1}{2} \sqrt{3}}.
\eeq
Since the Casimir energy is bounded, the asymptotic behaviour of $H_2$
as $\lambda \to \infty$ is the same as $H_1.$

As $\lambda \to \infty$ the $L^2$ finite condition (\ref{L2finite})
again leads to two asymptotic behaviours,
$$
f^2 {\rm vol} \sim  {\tilde c}_1 \log(\lambda) \lambda^{-4p-5} 
\quad {\rm and} \quad
f^2 {\rm vol} \sim  {\tilde c}_2 \log(\lambda) \lambda^{4p-5},
$$
for suitable constants ${\tilde c}_1$ and ${\tilde c}_2$ in the case
$p > 0.$  For $p=0,$ the curvature term has an influence on the
asymptotic behaviour. For $H_0,$ we obtain
$$
f^2 {\rm vol} \sim  {\tilde c}_1 \log(\lambda) \lambda^{-5} 
\quad {\rm and} \quad
f^2 {\rm vol} \sim  {\tilde c}_2 \log(\lambda)^3 \lambda^{-5},
$$
whereas for $H_1$ and $H_2$ the asymptotic behaviour is 
$$
f^2 {\rm vol} \sim  {\tilde c}_1 \log(\lambda)^{2 -
  \sqrt{3}} \lambda^{-5}  
\quad {\rm and} \quad
f^2 {\rm vol} \sim  {\tilde c}_2 \log(\lambda)^{2 +
  \sqrt{3}} \lambda^{-5}. 
$$
Hence for $p < 1$ both solutions lead to $L^2$ finite solutions,
resulting in a limit-circle case. 

In summary, for all values of $k$ and $s$ the boundary
conditions are either of limit-circle or of limit-point type. 
Furthermore, both end points are non-oscillatory and independent of
the energy eigenvalue $E.$ This allows us to apply theorem
7.5 in \cite{Pryce} which ensures that the spectrum is purely
discrete, bounded below and unbounded above. 

If limit-circle endpoints are present, a Sturm-Liouville problem is not
self-adjoint unless boundary conditions are imposed. For a
non-oscillatory endpoint $a$ there is always one solution which is
``small'', known as the subdominant solution. More precisely, the 
subdominant solution $u(\lambda),$ unique up to a scalar factor, satisfies
\beq
\lim\limits_{\lambda \to a} \frac{u(\lambda)}{v(\lambda)} = 0,
\eeq
where $v(\lambda)$ is any linearly independent solution, see theorem
7.15(i) in \cite{Pryce}. In our case, the subdominant solution is the
non-singular solution, e.g.\ $c_1 \lambda^l$ as $\lambda \to 0.$ A
natural boundary condition is defined by the subdominance condition,
i.e.\ always taking the non-singular solution. Theorem 7.21 in
\cite{Pryce} then guarantees that the subdominance condition defines a
valid self-adjoint problem known as the Friedrichs extension. 

\subsection{Numerical results}

In the following, we briefly sketch our numerical scheme for $H_0$
with $k$ and $s$ fixed.
The equations for $H_1$ and $H_2$ require minor, but straightforward
modifications. We are using a multi-component shooting method which
solves  a collection of initial
value problems (for the eigenvalue equation at fixed $E$) numerically with
a standard adaptive Runge-Kutta method. The first initial value problem
has initial values at $\lambda_0\approx 0$ and the solution $a(\lambda)$
of the
differential equation (\ref{Heq}) is evaluated at a matching point
$\lambda_m \gg  \lambda_0$. Allowing
the initial data to span the $m=2\min(k,s) +1$ dimensional space
specified by the boundary condition at $0$, this produces
a $2m\times m$ matrix $\phi$ whose columns consist of 
$a(\lambda_m)$ above $a'(\lambda_m)$.
The second initial value problem has
initial values at $\lambda_\infty \gg \lambda_m$ and is again
evaluated at $\lambda_m,$ resulting, as the initial
data span the boundary condition, in a $2m\times m$ matrix $\eta$
constructed similarly.
By construction, a general solution satisfying the boundary condition at
$0$ lies, at $\lambda_m$,
 in the column span of $\phi$, while a general solution
satisfying the boundary condition at $\infty$ lies, at $\lambda_m$ in the 
column span of $\eta$. Hence, $E$ is an eigenvalue if these spans have
nontrivial intersection, that is,
if
\beq
\label{deteq}
d(E)=\det(\phi \eta)=0.
\eeq 
Having constructed $d(E)$ numerically, we find its roots using a bisection
method. Typical values of our constants are $\lambda_0 =
0.001,$ $\lambda_m = 3$ and $\lambda_\infty = 30.$

\begin{table}[!ht]
\begin{center}
\begin{tabular}{|c|c|c||c|c|c||c|c|c|}
\hline
\multicolumn{3}{|c||}{\rule[-3mm]{0mm}{10mm} $H_0 = \tfrac{1}{2}
  \triangle$} &  
\multicolumn{3}{|c||}{$H_1 = \tfrac{1}{2} \triangle + \tfrac{1}{4} \kappa$} &
\multicolumn{3}{|c|}{$H_2 = \tfrac{1}{2} \triangle + \tfrac{1}{4}
  \kappa + \Cas$} \\
\hline
\rule[-3mm]{0mm}{8mm} energy & degeneracy & $\{ k, s \}^P$ 
&
energy & degeneracy & $\{ k, s \}^P$ 
& 
energy & degeneracy & $\{ k, s \}^P$ \\
\hline
0.00 & 1 & $\{0,0\}^+$ & 0.22 &  1 & $\{0,0\}^+$ &
% 0.58 & 1 & $\{0,0\}^+$ \\
-0.41 & 1 & $\{0,0\}^+$ \\
0.13 & 6 & $\{0,1\}^-$ & 0.38 & 6 & $\{0,1\}^-$ & 
% 1.89 & 6 & $\{0,1\}^-$ \\
-0.26$^*$ & 9 & $\{1,1\}^+$ \\
0.18 & 9 & $ \{1,1\}^+$ & 0.40 & 9 & $ \{1,1\}^+$ & 
% 1.93 & 9 & $ \{1,1\}^+$ \\
-0.19$^*$ & 6 & $ \{0,1\}^-$ \\
0.29 & 1 & $\{0,0\}^+$ & 0.53 & 1 & $\{0,0\}^+$ & 
% 2.93 & 1 & $\{0,0\}^+$ \\
-0.12 & 1 & $\{0,0\}^+$ \\
0.34 & 9 & $ \{1,1\}^+$ & 0.58 & 9 & $ \{1,1\}^+$ & 
% 3.49 & 9 & $ \{1,1\}^+$ \\
-0.01$^*$ & 9 & $\{1,1\}^-$ \\
0.35 & 10 & $\{0,2\}^+$ & 0.59$^*$ & 9 & $ \{1,1\}^-$ & 
% 3.51$^*$ & 9 & $ \{1,1\}^-$  \\ % switch order
0.00$^*$ & 9 & $ \{1,1\}^+$ \\
0.38 & 9 & $ \{1,1\}^-$ & 0.61$^*$ & 10 & $\{0,2\}^+$ & 
% 3.76$^*$ & 10 & $\{0,2\}^+$\\
0.01$^*$ & 25 & $\{2,2\}^+$ \\
0.40 & 30 & $\{1,2\}^-$ & 0.63 & 30 & $\{1,2\}^-$ & 
% 3.95 & 30 & $\{1,2\}^-$ \\
0.06 & 30 & $ \{1,2\}^-$ \\
0.49 & 25 & $\{2,2\}^+$ & 0.70 & 25 & $\{2,2\}^+$ & 
% 4.25 & 25 & $\{2,2\}^+$ \\
0.06$^*$ & 10 & $ \{0,2\}^+$ \\
0.54 & 6 & $\{0,1\}^-$ & 0.80 & 6 & $\{0,1\}^-$ & 
% 5.07 & 6 & $\{0,1\}^-$ \\
0.15 & 6 & $\{0,1\}^-$ \\
0.62 & 9 & $ \{1,1\}^+$ & 0.83 & 9 & $ \{1,1\}^+$ & 
% 5.39 & 9 & $ \{1,1\}^+$ \\
0.22 & 9 & $\{1,1\}^+$ \\
0.63 & 30 & $\{1,2\}^-$ & 0.87 & 30 & $\{1,2\}^-$ & 
% 5.81 & 30 & $\{1,2\}^-$  \\
0.29 & 30 & $\{1,2\}^-$ \\
0.64 & 14 & $\{0,3\}^-$ & 0.88$^*$ & 30 & $\{1,2\}^+$ &
% 5.91$^*$ & 30 & $\{1,2\}^+$ \\ % switch order
0.32$^*$ &  30 & $\{1,2\}^+$ \\
0.67 & 30 & $\{1,2\}^+$ & 0.91$^*$ & 14 & $\{0,3\}^-$ &
% 6.22$^*$ & 14 & $\{0,3\}^-$  \\
0.33$^*$ & 25 & $\{2,2\}^-$ \\
0.68 & 25 & $\{2,2\}^+$ & 0.93 & 25 & $\{2,2\}^+$ & 
% 6.24 & 25 & $\{2,2\}^+$ \\
0.34 & 25 & $\{2,2\}^+$ \\
0.69 & 42 & $ \{1,3\}^+$ & 0.95 & 42 & $ \{1,3\}^+$ & 
% 6.44$^*$ & 25 & $\{2,2\}^-$ \\ % switch (Casimir only)
0.37$^*$ & 14 & $\{0,3\}^-$ \\
0.76 & 25 & $\{2,2\}^-$ & 0.96 & 25 & $\{2,2\}^-$ & 
% 6.48$^*$ & 42 & $ \{1,3\}^+$ \\
0.40$^*$ & 42 & $\{1,3\}^+$ \\
0.79 & 70 & $\{2,3\}^-$ & 1.01 & 70 & $\{2,3\}^-$ & 
% 6.96 & 70 & $\{2,3\}^-$  \\
0.41$^*$ & 1 & $\{0,0\}^+$ \\
0.81 & 1 & $\{0,0\}^+$ & 1.07 & 1 & $\{0,0\}^+$ & 
% 7.22 & 1 & $\{0,0\}^+$ \\
0.41$^*$ & 49 & $\{3,3\}^+$ \\
\hline
\end{tabular}
\caption{This table shows the lowest energy eigenvalues of $H_0,$
  $H_1$ and $H_2.$ Each energy eigenvalue is labelled with its
  degeneracy and its quantum numbers $k$ and $s$ and parity $P.$
Energy eigenvalues marked with a 
  $*$ occur in a different order than their counterpart in the $H_0$
  spectrum. 
\label{energies}
}
\end{center}
\end{table}

Table \ref{energies} shows the lowest energy levels for $H_0 =
\tfrac{1}{2} \triangle$
which corresponds to the Laplacian for the $L^2$ metric, for $H_1 =
\tfrac{1}{2} \triangle + \tfrac{1}{4} \kappa$ which also includes the
curvature corrections, and for $H_2 = \tfrac{1}{2} \triangle +
\tfrac{1}{4}\kappa + \Cas$ which includes curvature corrections 
and Casimir energy. Energy levels are rounded to two decimal places.
The energy levels are ordered according to their respective energy,
and are labelled by isospin quantum number $k,$ total angular
momentum quantum number $s$ and parity $P.$ Curly brackets indicate that
the states $(k,s)$ and $(s,k)$ have the same energy. States with the
same quantum numbers but different energies form sequences of radially
excited states. We have chosen to display levels with
an energy not greater than the second excited $(k=0,s=0)$ state,
which are 19 energy levels in total. The columns ``degeneracy'' give
the number of different states with the same energy. For $k=s,$ the
degeneracy is given by $(2k+1)^2,$ whereas for $k\neq s$ it is
$2(2k+1)(2s+1),$ where the extra factor of $2$ arises from the $(k,s)
\to (s,k)$ symmetry of the spectrum.

The spectra of $H_0$ and $H_1$ are remarkably similar. The order of the 
energy levels remains the same apart from two exceptions which are 
marked with a $*$ in table \ref{energies} and which will be discussed in 
more detail later. The Casimir energy leads to significant changes, 
nevertheless the spectrum still shares some similarities.
The ground state of $H_0$ has energy $E_0^{(0)}= 0.00.$ The curvature
term increases the energy of the ground state of $H_1$ to $E_0^{(1)} =
0.22,$ whereas the renormalised Casimir energy leads to a decrease in
the ground state of $H_2$ to $E_0^{(2)} = -0.41.$  

In order to compare the energies of the excited states, we shift the
spectrum of the ground states of $H_1$ and $H_2$ to $0.00$ and then
calculate the relative difference of the corresponding energy
levels. Hence, the relative difference in energy between the $n$th excited
state $E_n^{(0)}$ of $H_0$ and the $n$th excited state $E_n^{(1)}$
of $H_1$ is given by 
\beq
\label{relative}
\frac{E_n^{(0)}-\left(E_n^{(1)}-E_0^{(1)}\right)}{E_n^{(0)}}.
\eeq
The relative difference of the $\{0,1\}^-$ states of $H_0$ and $H_1$ is
$15\%,$ and $7.7\%$ for the first excited $\{0,0\}^+$ states. All the
remaining states have a relative difference of less than $7\%.$ 

As mentioned earlier, some energy
levels change order. Therefore, it is useful to calculate the relative
error of states with the same quantum numbers $\{k, s\}^P.$ The first
transposition occurs for the first $\{0,2\}^+$ state and the first 
$\{1,1\}^-$ state. The relative difference between the energy of the
$\{0,2\}^+$ states of $H_0$ and $H_1$ is $10.4\%.$ The second 
transposition occurs for the first $\{0,3\}^-$ and
the first even $\{1,2\}^+$ state, and the $\{0,3\}^-$ states have a 
relative differences of $8.0\%.$ This leads to the observation that
the difference between positive and negative parity states with the
same $\{k,s\}$ is reduced for $H_1$ compared to $H_0$. Furthermore, $H_1$ 
seem to favour states with $k \approx s.$ 

As can be seen in table \ref{energies} the spectrum of $H_2$ shows many 
transpositions compared to the spectra of $H_0$ and $H_1.$ However, these 
transpositions occur for states which are close in energy, and the 
relative positions of the 
$\{0,0\}^+$ state and the $\{0,0\}^+$ excited states remain almost
unchanged. Calculating the relative differences as in formula 
(\ref{relative}) shows that all relative differences are less than 
$24\%.$ 

\subsection{Changing the radii}\label{scaling}

It is interesting to consider how our results change if the radii $R_1,R_2$
of the domain and target spheres are altered. 
We adopt the convention that a tilde signifies the case of general $R_1,R_2$
while undecorated variables refer to the case $R_1=R_2=1$. 
It is immediate from 
(\ref{marmes}) that the $L^2$ metric on $\rat_n$ scales as $\wt\gamma=
R_1^2R_2^2\gamma$, and hence, from (\ref{lap}), we see that the Laplacian
scales as $\wt\triangle=(R_1R_2)^{-2}\triangle$. Now scalar curvature
scales in the same way as the Laplacian (as can easily be seen, in this case,
from the formula for $\kappa$ in \cite{spe3}, for example). Hence, 
the Hamiltonians $H_0$ and $H_1$,  scale homogeneously,
$\wt{H}_i=(R_1R_2)^{-2}H_i$, and so their spectra can be obtained from
table \ref{energies} by a simple rescaling. Now the energies shown in this
table are quantum {\em corrections} to the classical energy, which, by 
the Lichnerowicz bound is $4\pi R_2^2$. Hence, for $H_0$, $H_1$, the
total energy of the $k$th energy level is
\beq
\wt{E}_k^{\rm total}=4\pi R_2^2+\frac{E_k}{R_1^2R_2^2}
\eeq
where $E_k$ is the eigenvalue of $H_i$. It is interesting to note that
the quantum correction becomes (naively) dominant in the case of small
target space ($R_2$ small).

The behaviour of $H_2$ is more subtle because the Casmir energy
scales differently from the other terms. To see this, we determine
how the Jacobi operator for a harmonic map $\phi:M\ra N$ scales under 
homotheties of $M$ and $N$.

\begin{prop}\label{homothety}
 Let $\vphi:(M^m,g)\ra (N^n,h)$ be harmonic with Jacobi
operator $\jac$. If $\wt{g}=R_1^2\, g$, $\wt{h}=R_2^2\, h$, where 
$R_1,R_2>0$ are constants, then the Jacobi operator of $\vphi$ as a
harmonic map $(M,\wt{g})\ra (N,\wt{h})$ is
$$
\wt{\jac}=R_1^{-2}\jac.
$$
\end{prop}

\begin{proof}
Let $\vphi_{s,t}:M\ra N$ be
a smooth two-parameter variation of $\vphi=\vphi_{0,0}$, with
$\cd_s\vphi_{s,t}|_{s=t=0}=X$, 
$\cd_t\vphi_{s,t}|_{s=t=0}=Y\in\Gamma(\vphi^{-1}TN)$. We have, in obvious
notation,
\bea
\wt{E}(\vphi_{s,t})&=&R_1^{m-2}R_2^2E(\vphi_{s,t})\nonumber \\
\Rightarrow\quad
\wt{\hess}(X,Y)&=&
\left.\frac{\cd^2\:\:}{\cd s\, \cd t}\wt{E}(\vphi_{s,t})\right|_{s=t=0}
=R_1^{m-2}R_2^2\, \hess(X,Y)\nonumber\\
&=&R_1^{m-2}R_2^2\int_Mh(X,\jac Y)\vol
=\int_M\wt{h}(X,R^{-2}_1\jac Y)\wt{\vol}. \nonumber
\eea
\end{proof}

Since $\Cas$ is (formally) a sum of {\em square roots} of eigenvalues
of $\jac$, we see that it scales as $\wt\Cas=R_1^{-1}\Cas$. Hence,
the scaling behaviour of $H_2$ is inhomogeneous, and, except in the
special case that $R_1R_2^{-2}=1$, the spectrum of $\wt{H}_2$
cannot be deduced directly from the spectrum of $H_2$. Since $\Cas$ is 
independent of the radius $R_2$ whereas the spectra of $H_0$ and $H_1$ 
scale like $1/R_2^2,$ the effect of the Casimir energy becomes less 
important for small $R_2.$ 
We have also numerically checked that a smaller radius of the target 
sphere reduces the effect of the Casimir energy term significantly. 
For example for $R_1=1$ and $R_2=\frac{1}{4},$ there is only one 
transposition, compared to the $H_1$ spectrum.

Throughout this paper we have used the convention that
  $\hbar=1.$ If we reinstate the parameter $\hbar$ then the classical
energy scales like $\hbar^0,$ the Casimir energy scales like $\hbar^1$ 
whereas the Laplacian and curvature terms scale like $\hbar^2.$ Hence,
$\hbar$ can be removed by redefining $R_1 \mapsto R_1/\hbar,$ so the
semi-classical limit $\hbar \to 0$ coincides (formally) with the
planar limit $R_1 \to \infty.$

\begin{figure}[!ht]
\begin{center}
\includegraphics[scale=0.8]{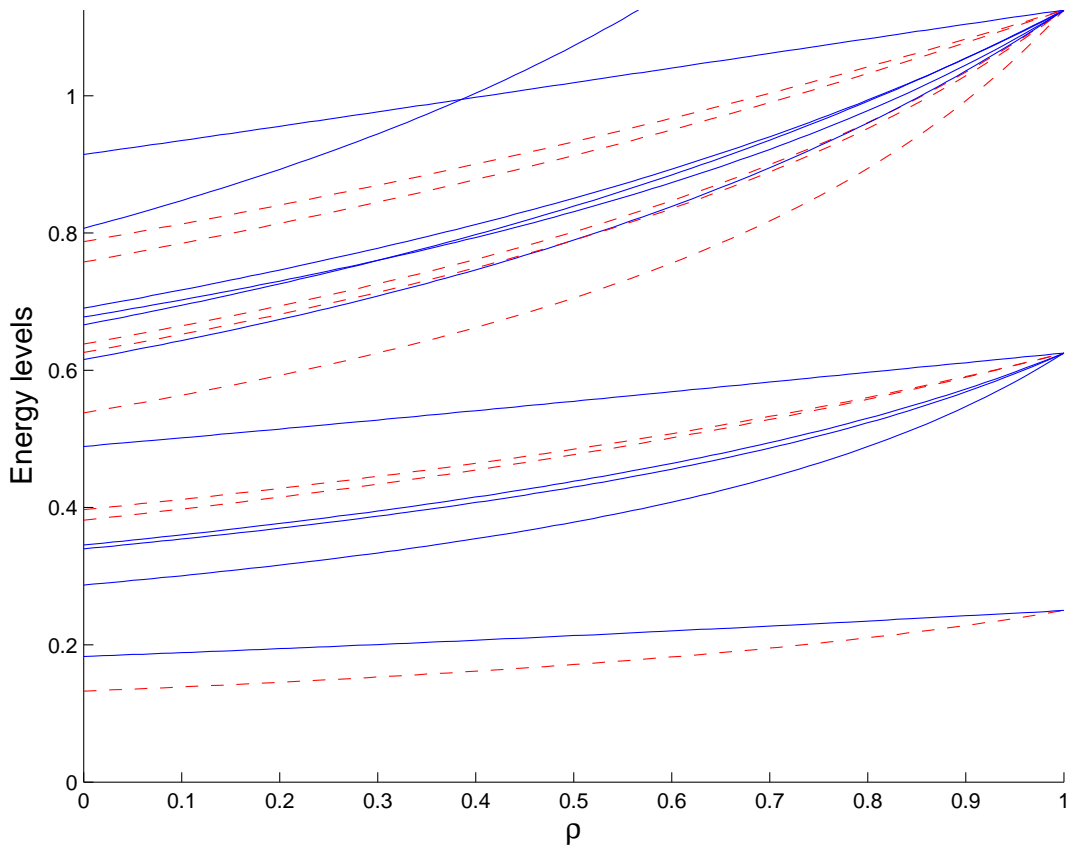}
\caption{Energy levels interpolating between the $L^2$ metric and the 
${\mathbb CP}^3$ metric. Here solid lines correspond to positive and
  dashed lines to negative parity.\label{energyflow}}
\end{center}
\end{figure}

\subsection{Deformation to the Fubini-Study metric}
As a nontrivial test of our calculations we calculate how the spectrum
changes as the metric is smoothly deformed from the $L^2$ metric
to the well-known Fubini-Study metric of ${\mathbb CP}^3,$ where the
spectrum and degeneracy of the Laplace operator have been calculated
explicitly. We consider the one-parameter family of
$SO(3)\times SO(3)$ invariant
k\"ahler metrics defined, as in Proposition \ref{metprop}, 
by the functions
\beq
\label{As}
A_\rho(\lambda) = 32 \rho A_{FS} (\lambda) + (1-\rho) A_{L^2}(\lambda),\qquad
0\leq \rho\leq 1, 
\eeq
where $A_{FS}$ and $A_{L^2}$ are given in (\ref{AFS}) and (\ref{Adef}),
respectively, and the factor of $32$ ensures that the eigenvalues of the 
two Laplacians are of the same order of magnitude. 
The $p$th 
eigenvalue of the Laplacian for the Fubini-Study metric on $\CP^3$
with constant holomorphic sectional curvature $4$ (and hence coefficient 
function $A_{FS}$) is 
\beq
%E_p = 4p(3+p)
E_p = 4p(p+3)
\eeq
with degeneracy
\beq
%  \deg (E_p) =  \tfrac{3}{4}(3+2p)(p+1)^2(p+2)^2,
  \deg (E_p) =  \tfrac{3}{4}(2p+3)(p+1)^2(p+2)^2,
\eeq
for $p\geq 1$ \cite{ber}. 
Although it is not made clear in \cite{ber}, the
degeneracy of $E_0=0$ must be $1$ by the Hodge isomorphism theorem, since
$\CP^3$ is connected, hence $\dim H^0(\CP^3)=1$ which equals the dimension 
of the space of harmonic $0$-forms. 
Note that $\deg(E_p)$ for $p\geq 1$ is an integer divisible by $9.$
In figure \ref{energyflow} we show how the energy levels change as the
parameter $\rho$ in (\ref{As}) is increased from $0$ (the $L^2$ metric)
to $1$ (the Fubini-Study metric). We
distinguish between states with even and odd parity. 
The ground state is $E_0=0$ for both
metrics. In figure \ref{energyflow} we follow all the energy levels of
table \ref{energies}, and it can be seen how the different energy
levels become degenerate at $\rho=1,$ the Fubini-Study limit. Our
numerically computed energy levels agree with the exact result to 
at least two decimal places, when we take into account the factor 
$\frac{1}{2}$ in $H_0$ and normalize the Fubini-Study energy by a
factor of $32,$ which arises from the factor $32$ in (\ref{As}). 

Our numerical scheme at $\rho=1$ does not find all eigenfunctions of
$\triangle$ on $\CP^3$, because, except for the ground state,
we impose a boundary condition
which forces the wavefunction to vanish on $\cd_\infty\rat_1$,
 the boundary of $\rat_1$
at infinity. It is known \cite{spe3} that $\cd_\infty\rat_1$ coincides with
the image of the inclusion $\CP^1\times\CP^1\hookrightarrow\CP^3$,
$([a_0,a_1],[b_0,b_1])\mapsto [a_0b_1,a_1b_1,a_1b_1,a_1b_0]$. 
The stabilizer of this subset of $\CP^3$ in $SU(4)$, the isometry group
of $\CP^3$, is isomorphic to $SU(2)\times SU(2)$ (the isometry
group of $\rat_1$ itself). So, given an eigenfunction vanishing
on $\cd_\infty\rat_1$, the $SU(4)$ action generates a 
$9$-dimensional orbit of degenerate eigenfunctions which do 
not vanish on $\cd_\infty\rat_1$ (because $\dim SU(4)=15$ and
$\dim SU(2)\times SU(2)=6$). Hence, we expect the degeneracies 
found by our numerics 
(at $\rho=1$)
to be a factor $9$ smaller than the
degeneracies of the true Laplacian on $\CP^3$.
This is precisely what we find. 
For example, the degeneracy of the second 
${\mathbb C}P^3$ eigenvalue is  $\deg (E_1)/9 = 15,$ which 
corresponds to the levels $\{0,1\}$ and $\{1,1\}$ with degeneracy
$2*3+3*3=15,$ see table \ref{energies}. 
Similarly, $\deg (E_2)/9 = 84,$ corresponds to $\{0,0\},$ $\{1,1\},$
  $\{0,2\},$ $\{1,1\},$ $\{1,2\},$ and $\{2,2\}$ with
  $1+3*3+2*5+3*3+2*3*5+5*5 =84.$ Further, $\deg (E_3)/9 = 300,$ which
corresponds to $\{0,1\},$ $\{1,1\},$ 
  $\{1,2\},$ $\{0,3\},$ $\{1,2\},$ $\{2,2\},$
 $\{1,3\},$ $\{2,2\},$ $\{2,3\},$ and $\{3,3\},$
whose degeneracies add up to $300,$ as expected. To display all the
levels which contribute to the fourth energy eigenvalue of 
${\mathbb C}P^3$ we have also included the $\{3,3 \}$ level in figure
\ref{energyflow}. Note that the second excited $\{0,0\}$ state
contributes to the fifth energy eigenvalue of ${\mathbb C}P^3.$

\section{Concluding remarks}\label{conclusion}
\news

In this paper we discussed the semi-classical quantization of soliton 
dynamics for ${\mathbb C}P^1$ lumps moving on a 2-sphere. We followed Moss 
and Shiiki \cite{mosshi} who derived a Born-Oppenheimer approximation to the 
quantum dynamics based on the moduli space approximation. 
We were able to evaluate three different truncations of the 
Born-Oppenheimer Hamiltonian $H_{BO},$ namely the geometric Laplacian $H_0 
= \frac{1}{2} \triangle,$ the first geometric correction $H_1$ given by 
the scalar curvature $\kappa$ of the moduli space, and the Hamiltonian $H_2$ 
which consists of $H_1$ together with the Casimir energy. The 
Casimir energy is notoriously difficult to evaluate. At $\lambda=0$
and $\lambda = \infty$ the spectrum of the Jacobi operator is known
explicitly. Our approach is to calculate the regularized Casimir
energy using zeta-function regularization for $\lambda = 0$ and
$\lambda = \infty.$ Then the intermediate values are calculated using
an approximate self-similarity of the eigenvalue spectrum. 
We calculated the first $19$ energy levels for $H_0,$ $H_1$ and $H_2,$
and found that the first two spectra are remarkably similar whereas the 
inclusion of the Casimir energy leads to significant changes. 
There is an overall shift of the spectrum of $H_1$ and $H_2.$ The relative 
errors between $H_0$ and $H_1$ are $15\%$ for the first excited state and
less that $8\%$ for other excited states. By contrast, the relative errors
between $H_0$ and $H_2$ are as high as $24\%$ for some excited states. 
There are only two transpositions of states where energy levels of $H_1$ 
do not have the same order compared to the $H_0$ spectrum. However, the 
spectrum of $H_2$ shows many transpositions.

We proved that all the spectra enjoy a rather surprising spin-isospin 
exchange symmetry. The proof rested on the identification of a hidden
isometry of $\rat_1$ which becomes manifest only after one identifies
$\rat_1\cong TSO(3)$. As a non-trivial check of our calculations we
calculated how the energy levels change when we interpolate between the
$L^2$ metric and the Fubini-Study metric on ${\mathbb C}P^3.$  We
reproduced the known spectrum for ${\mathbb C}P^3$ and discussed the
degeneracies of the energy levels.

Our approach allows us to calculate the spectrum of the Laplacian for
any $SO(3) \times SO(3)$ invariant k\"ahler metric on $\rat_1$.
One interesting choice is $A = c/\Lambda,$ for which 
the coefficient of 
$(\lamvec \cdot {\bf S})^2$ vanishes in (\ref{lapnew}). In this case,
the angular momentum operator ${\bf L}^2$ also commutes with the
Hamiltonian and $l$ becomes an additional quantum number. 

The lump dynamics serves as a toy model for other solitons. For
example in three spatial dimensions, two physically relevant models are
the Skyrme model \cite{skyrme} and the Faddeev-Hopf model
\cite{fadd, faddnie}. In both models, the solitons can be quantized as
fermions due to so-called Finkelstein-Rubinstein constraints
\cite{finrub}. Adkins, Nappi and Witten first quantized the $B=1$ Skyrmion 
in \cite{Adkins} where $B$ is the topological charge. The effects of the 
Casimir energy on the predictions of the Skyrme model in the $B=1$
sector have   
been discussed in great detail in \cite{Meier}. The authors calculated the 
1-loop corrections to various physical quantities using phase-shift 
techniques to evaluate the Casimir energy. For higher topological charge 
the ground and lowest excited states in the 
Skyrme model and the Faddeev-Hopf model were calculated in 
\cite{Irwin, kru, Battye} and \cite{kruspe} using zero-mode quantization 
and Finkelstein-Rubinstein constraints. The results of this paper suggest 
that the order of the states would remain the same, with minor changes, if 
the curvature correction in the Born-Oppenheimer approximation was taken 
into account. However, the Casimir energy could lead to significant 
changes in the order of states. A more careful analysis of higher order 
terms in these models, and in particular, of the Casimir energy, would be 
very useful.

\section*{Acknowledgments}

SK is grateful to NS Manton for interesting discussions at various
stages of the project. The authors would like to thank James McGlade
and Lamia Al Qahtani for useful discussions. McGlade derived the
spectrum for zero isospin in his thesis \cite{McGlade}.

\section*{Appendix}
\appendix

\section{The rough Laplacian}

In this appendix we will compute in detail the action of the rough Laplacian
$\triangle_\vphi$ on sections $Y$ of $\vphi^{-1}TN$ of the form
\beq
Y=a(\theta)\cos m\phi\wt{E}_1.
\eeq
Recall that $(M,g)=(N,h)=$ the unit sphere, $(\theta,\phi)$ are the usual
polar coordinates on $M$ or $N$, $E_1=\cd/\cd\theta$, $E_2=\cosec\theta\cd/\cd
\phi$ is an orthonormal frame on $M$ or $N$, $\wt{E}_1=E_1\circ\vphi$,
$\wt{E}_2\circ\vphi$, and $\vphi:M\ra N$ is a map of the form
$\vphi(\theta,\phi)=(f(\theta),\phi)$ in polar coordinates. Now, 
given a map between Riemannian manifolds $\vphi:M\ra N$, the pullback
connexion $\nabla^\vphi$ is the unique connexion on the vector bundle
$\vphi^{-1}TN$ satisfying the axioms
\bea
\nabla^\vphi_X(Y_1+Y_2)&=&\nabla^\vphi_XY_1+\nabla^\vphi_XY_2\label{a1}\\
\nabla^\vphi_X(fY_1)&=&X[f]Y_1+f\nabla^\vphi_XY_1\label{a2}\\
\nabla^\vphi_X(Y\circ\vphi)&=&(\nabla^N_{\d\vphi X}Y)\circ\vphi\label{a3}
\eea
for all $X\in\Gamma(TM)$, $Y_1,Y_2\in\Gamma(\vphi^{-1}TN)$, $f\in C^\infty(M)$
and $Y\in\Gamma(TN)$, where $\nabla^N$ is the Levi-Civita connexion on
$(N,h)$. In our case $\nabla^N$ (which coincides with $\nabla^M$) is
determined by its action on the frame $\{E_1,E_2\}$,
\beq\label{a4}
\nabla^NE_1=\cot\theta\, e_2\otimes E_2,\qquad
\nabla^NE_2=-\cot\theta\, e_2\otimes E_1,
\eeq
where $\{e_1,e_2\}$ is the coframe dual to $\{E_1,E_2\}$. From this we may 
deduce how $\nabla^\vphi$ acts on $\wt{E}_1,\wt{E}_2$. By property
(\ref{a3}) and (\ref{a4}),
\bea
\nabla^\vphi_{E_1}\wt{E}_1&=&(\nabla^N_{f'(\theta)E_1}E_1)\circ\vphi=0\\
\nabla^\vphi_{E_1}\wt{E}_2&=&(\nabla^N_{f'(\theta)E_1}E_2)\circ\vphi=0\\
\nabla^\vphi_{E_2}\wt{E}_1&=&
(\nabla^N_{\frac{\sin f}{\sin\theta}E_2}E_1)\circ\vphi 
=\frac{\cos f}{\sin\theta}\wt{E}_2\\
\nabla^\vphi_{E_2}\wt{E}_2&=&
(\nabla^N_{\frac{\sin f}{\sin\theta}E_2}E_2)\circ\vphi 
=-\frac{\cos f}{\sin\theta}\wt{E}_1.
\eea

The rough Laplacian has 4 terms,
\beq
\triangle_\vphi Y=-\nabla^\vphi_{E_1}(\nabla^\vphi_{E_1}Y)
-\nabla^\vphi_{E_2}(\nabla^\vphi_{E_2}Y)
+\nabla^\vphi_{\nabla^M_{E_1}E_1}Y
+\nabla^\vphi_{\nabla^M_{E_2}E_2}Y.
\eeq
We evaluate these in turn, for $Y=a(\theta)\cos m\phi\wt{E}_1$. First
\beq
\nabla^\vphi_{E_1}(\nabla^\vphi_{E_1}a\cos m\phi \wt{E}_1)=
\frac{\cd^2\:}{\cd\theta^2}(a\cos m\phi)\wt{E}_1=a''\cos m\phi\wt{E}_1.
\eeq
The second term is more involved: 
\bea
\nabla^\vphi_{E_2}(a\cos m\phi\wt{E}_1)&=&aE_2[\cos m\phi]\wt{E}_1+
a\cos m\phi\nabla^\vphi_{E_2}\wt{E}_1\nonumber \\
&=&-ma\frac{\sin m\phi}{\sin\theta}\wt{E}_1+
a\frac{\cos m\phi\cos f}{\sin\theta} 
\wt{E}_2\\
\Rightarrow
\nabla^\vphi_{E_2}(\nabla^\vphi_{E_2}a\cos m\phi\wt{E}_1)&=&
-m^2a\frac{\cos m\phi}{\sin^2\theta}\wt{E}_1
-2ma\frac{\sin m\phi\cos f}{\sin^2\theta}\wt{E}_2\nonumber\\
&&-a\frac{\cos m\phi\cos^2 f}{\sin^2\theta}\wt{E}_1.
\eea
The third term vanishes since $\nabla^M_{E_1}E_1=0$, while the fourth term is
\beq
\nabla^\vphi_{\nabla^M_{E_2}E_2}(a\cos m\phi \wt{E}_1)=-\cot\theta
\nabla^\vphi_{E_1}(a\cos m\phi \wt{E}_1)=-a'\cot\theta\cos m\phi\wt{E}_1.
\eeq
Assembling the pieces, one sees that
\bea
\triangle_\vphi(a\cos m\phi\wt{E}_1)&=&\left(-a''-a'\cot\theta
+\frac{m^2}{\sin^2\theta}a+\frac{\cos^2 f}{\sin^2\theta}a\right)
\cos m\phi\wt{E}_1\nonumber \\
&&+\left(\frac{2m\cos f}{\sin^2\theta}a\right)\sin m\phi\wt{E}_2,
\eea
as claimed in equation (\ref{raaba}).

\section{Angular momentum calculations}
\label{AngularMomentum}
\news

In this appendix we briefly describe the evaluation of the operator 
$(\lamvec \cdot {\bf S})^2.$ We closely follow the notation in \cite{Edmonds}.
Note the spin quantum number $j$ satisfies $j=k$ since ${\bf J}^2={\bf 
K}^2$.
We choose the following convention for the spherical harmonics, 
\beq
Y_{l,m}(\theta, \phi) = (-1)^m
\sqrt{ \frac{(2l+1)}{4 \pi} \frac{(l-m)!}{(l+m)!}}
{P_l}^m(\cos \theta) \exp(i m \phi),
\eeq
where the associated Legendre polynomials are given by
\beq
{P_l}^m(x) = \frac{(1-x^2)^{m/2}}{2^l l!} \frac{d^{l+m}}{dx^{l+m}} 
(x^2-1)^l.
\eeq
Note that $Y^*_{l,m} = (-1)^m Y_{l,-m}$ and
\beq
\int\limits_{0}^{2 \pi} \int\limits_{0}^{\pi}
Y^*_{l,m}(\theta,\phi)  Y_{l^\prime,m^\prime}(\theta, \phi) 
\sin \theta\, {\rm d}\theta {\rm d}\phi 
= \delta_{ll^\prime} \delta_{m m^\prime}.
\eeq

It is easy to show that $(\lamvec \cdot {\bf S}) = (\lamvec \cdot {\bf 
J}),$ and we showed in section \ref{LRat1} that this operator commutes 
with ${\bf S}$ and ${\bf J}^2,$ see equation (\ref{Xcommutes}) 
where ${\bf J} = -i \thetavec$ and ${\bf S} = -i {\bf X}.$ Hence, the
operator does not change the quantum numbers $s,$ $k,$ $k_3$ and
$s_3.$ Furthermore, the matrix elements are independent of $s_3$ and
$k_3,$ hence we can choose without loss of generality $s_3=0$ and
$k_3=0$. We need to evaluate 
\beq
\label{M1}
M \equiv \left(M_{{\tilde l}l}\right) = 
\langle k,{\tilde l}, s, 0 | 
\lamvec \cdot {\bf J} 
| k,l,s,0 \rangle ,
\eeq
which is a $(2 \min(k,s) +1)$ by $(2 \min(k,s) + 1)$ matrix, because the 
angular momentum quantum numbers satisfy $\min(k,s) \le l,{\tilde l} \le 
k+s$. Since $\lamvec \cdot {\bf J}$ is a Hermitian operator the matrix
$M$ is Hermitian, namely 
\beq
\label{herm}
M_{{\tilde l}l} = M_{l{\tilde l}}^\dagger.
\eeq
Using Clebsch-Gordan coefficients (\ref{M1}) can be rewritten as
\beq
M = \sum\limits_{m_1,m_2} 
(k, {\tilde l}, s, 0 | k, m_1, {\tilde l}, -m_1)
\langle k, m_1, {\tilde l}, -m_1|\lamvec \cdot {\bf J} 
| k, m_2, l, -m_2\rangle (k,m_2,l,-m_2|k, l, s, 0).
\eeq
The operator $\lamvec\cdot {\bf J}$ can be written in terms of spherical 
harmonics as
\beq
\lamvec\cdot {\bf J} = 
\frac{2 \sqrt{\pi}\lambda}{\sqrt{3}}
\left(
-\frac{1}{\sqrt{2}} Y_{1,1} J_- + \frac{1}{\sqrt{2}} Y_{1,-1} J_+ 
+ Y_{1,0} J_3
\right).
\eeq
Hence, we only need to evaluate
\bea
\langle k,m_1|J_+|k, m_2\rangle &=&
\delta_{m_1,m_2+1} \sqrt{(k-m_2)(k+m_2+1)}, \\
\langle k,m_1|J_-|k, m_2\rangle &=&
\delta_{m_1,m_2-1} \sqrt{(k+m_2)(k-m_2+1)}, \\
\langle k,m_1|J_3|k, m_2\rangle &=&
\delta_{m_1,m_2} m_2,
\eea
and
\beq
\langle {\tilde l},m_1 | Y_{l,m} | l,m_2 \rangle,
\eeq
which can be calculated using (4.6.3) in \cite{Edmonds}, namely,
\bea
&&\int\limits_0^{2\pi} \int\limits_0^\pi
Y_{l_1,m_1}(\theta, \phi) 
Y_{l_2,m_2}(\theta, \phi) 
Y_{l_3,m_3}(\theta, \phi) 
\sin\theta {\rm d} \theta {\rm d} \phi 
\\
\label{3jeq}
&&= 
\sqrt{\frac{(2l_1+1)(2l_2+1)(2l_3+1)}{4\pi}}
\left(
\begin{array}{ccc}
l_1&l_2&l_3\\
0&0&0
\end{array}
\right)
\left(
\begin{array}{ccc}
l_1&l_2&l_3\\
m_1&m_2&m_3
\end{array}
\right),
\eea
where the two matrices in (\ref{3jeq}) 
correspond to $3j$ symbols. The $3j$ symbol 
\beq
\left(
\begin{array}{ccc}
{\tilde l}&l&1\\
0&0&0
\end{array}
\right)
\eeq
vanishes unless ${\tilde l} = l \pm 1,$ which can be used to show that 
$M$ vanishes in the cases $k=0$ or $s=0$. For example, for $k=0$ the 
only allowed values of the angular momentum are $l=s$ and ${\tilde l} = 
s,$ which vanishes since ${\tilde l} \neq l \pm 1.$  

Following the
notation in \cite{Edmonds}, the Clebsch-Gordan coefficients and the
$3j$ symbols have real entries. Hence, $M$ is a real matrix and
(\ref{herm}) implies that $M$ is symmetric. 
Note that the eigenvalues of 
$M$ are $0, \pm 1, \pm 2, \dots, 
\pm \min(k,s),$ which can be explained as follows. 
Heuristically, $(\lamvec \cdot {\bf S})$ is the projection of
the total angular momentum operator onto the $\lamvec$
direction. Hence, it can be rotated by a change of variables to $S_3$
which has eigenvalues $0, \pm 1, \dots \pm \min(k,s).$ The highest 
possible value of $s_3$ is $\min(k,s)$ because the dimension
of $M$ is $2 \min(k,s) + 1.$  

Now, we can evaluate the matrix  (\ref{M1}) for various values of $k$
and $s$ using Maple, for example. For small values the formulae are 
more tractable. For example for $(k, s) =  (1,s)$ we obtain
\beq
\label{M2ex}
M^2 = 
\langle 1,{\tilde l}, s, 0 | 
(\lamvec \cdot {\bf J})^2
| 1,l,s,0 \rangle =
\frac{\lambda^2}{2s+1}
\left(
\begin{array}{ccc}
s+1 & 0     & \sqrt{s(s+1)} \\
0   & 2s+1 & 0\\
\sqrt{s(s+1)} &  0 & s
\end{array}
\right).
\eeq
We showed in section \ref{interchange} that there is a $(k,s) \to (s,k)$ 
symmetry for the respective matrices $M.$ This was verified by explicit 
calculation 
for all  values $(k,s)$ in table \ref{energies}.
As can be seen in (\ref{M2ex}) there is a natural chess board
structure in the matrix $M^2,$ which follows from the fact that $M$ 
vanishes unless ${\tilde l} = l \pm 1,$ and hence $M^2$ vanishes 
unless ${\tilde l} = l+2,$ $l,$ or $l-2.$

\end{document}